\documentclass[a4paper,10pt]{article}
\usepackage[top=3cm,bottom=3cm,left=1.5cm,right=1.5cm]{geometry}
\usepackage{multirow}
\usepackage{amsmath}
\usepackage[psamsfonts]{amssymb}
\usepackage{amsxtra}
\usepackage{threeparttable}

\usepackage[ruled,vlined, onelanguage]{algorithm2e}

\usepackage{hyperref}
\usepackage{graphicx}
\usepackage{caption}
\usepackage{subcaption}
\usepackage{amssymb}
\usepackage{color}
\usepackage{enumerate}
\usepackage{mathrsfs}
\usepackage{multicol}
\usepackage{amsopn}
\usepackage{amsthm}
\usepackage{amssymb}
\newcommand{\C}{\mathbb{C}}

\newcommand{\K}{\mathbb{K}}
\newcommand{\N}{\mathbb{N}}

\newcommand{\R}{\mathbb{R}}

\newcommand{\Z}{\mathbb{Z}}

\newcommand{\calA}{\mathcal{A}}

\newcommand{\calC}{\mathcal{C}}

\newcommand{\calE}{\mathcal{E}}
\newcommand{\calF}{\mathcal{F}}
\newcommand{\calH}{\mathcal{H}}

\newcommand{\calM}{\mathcal{M}}
\newcommand{\calP}{\mathcal{P}}

\newcommand{\calS}{\mathcal{S}}

\newcommand{\vphi}{\varphi}

\newcommand{\abs}[1]{\vert #1 \vert}
\newcommand{\norm}[1]{\Vert #1 \Vert}
\newcommand{\abss}[1]{\left\vert #1 \right\vert}

\newcommand{\set}[1]{\left\lbrace #1\right\rbrace}
\newcommand{\sse}{\subseteq}
\newcommand{\sprod}[1]{\left\langle #1 \right\rangle}

\usepackage{dsfont}

\newcommand{\geqsim}{\gtrsim}
\newcommand{\leqsim}{\lesssim}

\newcommand{\st}{\text{subject to}}

\DeclareMathOperator{\supp}{supp}
\DeclareMathOperator{\id}{id}

\DeclareMathOperator{\ran}{ran}

\newcommand{\re}{\text{ Re }}

\newcommand{\hatphi}{\widehat{\phi}}
\newcommand{\hatsigma}{\widehat{\sigma}}

\newtheorem{lem}{Lemma}
\newtheorem{prop}[lem]{Proposition}

\newtheorem{theo}[lem]{Theorem}
\newtheorem{cor}[lem]{Corollary}

\newtheorem{rem}[lem]{Remark}

\numberwithin{lem}{section}
 
\title{Thermal Source Localization Through Infinite-Dimensional Compressed Sensing}

\author{Axel Flinth \and
Ali Hashemi }

\date{ Institut f\"ur Mathematik,Technische Universit{\"a}t Berlin. \\
E-mail: $\{$flinth,hashemi$\}$@math.tu-berlin.de \\
\vspace{.2cm}
\today}
\begin{document}
 
\maketitle
\abstract{ We propose a scheme utilizing ideas from infinite dimensional compressed sensing for thermal source localization. Using the soft recovery framework of one of the authors, we provide rigorous theoretical guarantees for the recovery performance. In particular, we extend the framework in order to also include noisy measurements. Further, we conduct numerical experiments, showing that our proposed method has strong performance, in a wide range of settings. These include scenarios with few sensors, off-grid source positioning and high noise levels, both in one and two dimensions.}

\vspace{-.1cm}

\section{Introduction}
 

Monitoring temperature over spatial domains is an important task with practical importance for surveillance and automation purposes \cite{2}. Areas of application include agriculture \cite{baviskar2014real}, climate change studies \cite{young2014low}, thermal monitoring of CPUs \cite{ranieri2015near}, and environmental protection \cite{3}.  Other applications include monitoring of the thermal field caused by on-chip sensors for many-core systems \cite{ranieri2015near,ranieri2014near}, and cultivation of sensitive species \cite{jiang2016wireless}. For these and similar  applications, it is necessary to monitor the temperature with high resolution using a proper sampling device for possibly long term periods. One possible way to do this is to estimate the positions of the initial heat sources.

\subsection{Previous Works and Literature Review}
 
Solving the inverse problems involving the heat or diffusion equation has a reach history coming back to estimate the initial temperature of the earth by Fourier and Kelvin \cite{ranieri2014sensing}. These inverse problems and strategies to solve them efficiently has attracted much attention recently, due to the numerous challenges that we are facing in real-world applications. First of all, we usually have a tight constraint on the number of sensors in practice, e.g. due to economical constraints. These sensors also usually have limited power supply as well. There is no need to mention that, the nature is not always cooperating with us; therefore, we  have noise and interference issues. These factors make the heat involving inverse problem hard to solve. Hence, we need to somehow incorporate inherent structure of thermal field  as a side information and utilize an intelligent machinery in order to tackle this tricky problem.

In the literature, many strategies for attacking the problem of thermal field and/or source estimation can be found.  Let us summarize a few of the more recent ones.

Regarding the mathematical analysis, the authors in \cite{dokmanic2011sensor,ranieri2013sampling} provide fundamental error bounds on the aliasing error of reconstructing a diffusion field. The main argument is that the diffusion process inherently acts as a low pass filter.  Hence, although the spatial bandwidth of the diffusion field is infinite, it can still be well approximated by a signal with low bandwidth. Considering this fact the authors proposed a method for reconstructing the thermal field generated by initial sources. This contribution is extended in \cite{ranieri2012sampling} for addressing the problem of time-varying sources by considering time varying emissions rates lying in two specific low dimensional sub-spaces. 

This effort has been extended by Murray et al. \cite{murray2015estimating, murray2014spatio}, exploiting Prony's method for localized sources in order to estimate the locations in space and time of multiple sources. In a very recent work, they extended their method to distributed sensor networks \cite{murray2016physics},and to non-localized sources of diffusion fields, in particular to straight line and polygonal sources \cite{murray2016reconstructing}, respectively. They furthermore consider other governing equations, such as the wave and Poisson equations \cite{murray2017universal}.  However, these works assume a relatively large amount of field samples to be accessible, which may not be possible in a wide range of applications. 


The works mentioned above do not address the problem of reducing the number of spatiotemporal samples in an efficient way. This is instead done in \cite{ranieri2015near,ranieri2014near}, where the problem  of low resolution in thermal monitoring of a CPU is considered. Their proposed method consists in selecting the most informative sensors utilizing a frame potential objective function based on the Unit-Norm-Tight-Frame (UNTF) concept. Although the paper emphasizes on empirical aspects, theoretical claims are also derived.

Excluding \cite{murray2015estimating, murray2014spatio}, all of the previously mentioned works do not explicitly try to utilize the useful structures which exist for diffusion sources, such as sparsity and spatiotemporal correlation among measurement governed by a PDE constraint. In particular, none of them use the powerful framework of \emph{compressed sensing} \cite{CandesTao2005}. Therefore, Rostami et al. \cite{rostami2013compressed} proposed a compressed sensing method for reconstructing the diffusion field by incorporating PDE constraints in recovery part as a side information. This effort was extended to the 2D scenario in \cite{hashemi2016efficient} by one of the authors of this article, together with co-authors. In a similar way but by utilizing an analysis formulation for the source localization problem, \cite{kitic2016physics} consider the source localization inverse problem for other types of sources or governing equations, with PDE side information in a co-sparse framework. 

\subsection{Contributions} 
The contribution of this work is threefold: As we have seen in the preceding section, rigorous theoretical guarantees for the methods described above are relatively scarce. Therefore, our first contribution in this paper is to provide mathematical guarantees and proofs rigorously showing that our recovery framework will provide solutions which, in a certain sense, are close to the ground truth solutions. This type of recovery, coined \emph{soft recovery}, was recently developed by one of the authors in \cite{flinth2017soft}. In this paper, we in fact extend the theoretical soft recovery framework to a setting with noisy measurements -- a mathematical contribution important on its own.

Secondly, the subtle problem of off-grid source positioning has not been considered so far in  the thermal source localization literature, to the best knowledge of the authors. The proposed schemes in the area assume a fixed grid, on which a discretized version of the continuous field is defined. A fixed grid does not only cause discretization errors, but it also has problems capturing sources which have positions between the grid points \cite{chi2011sensitivity}. In this work, we address this issue by employing the philosophy of \cite{tang2013compressed, candes2014towards,DuvalPeyre2015}. That is, we consider the recovery problem in a continuous domain. This approach automatically tackles the discretization error and off-grid challenge, since it in some sense uses an infinite dimensional grid. 
 
The final contribution of this paper consists in its robustness toward the insufficient number of spatiotemporal measurements and also ill-conditioning of the Green function matrix. While most proposed methods in the literature assume a large number of accessible measurements, the numerical results based on the idea proposed in this paper demonstrate a very satisfactory performance, already when we only have access to one time sample and a very small subset of spatial samples. These results could have a significant impact on applications which have a tight constraint on the number of spatiotemporal samples due to financial costs, energy consumption, physical limitation or other application specific issues. Besides, the recovery procedure is much more robust to the ill-conditioning of the system model matrix consisting of Green functions compared to the classical results in the literature. In Section \ref{sec:Numerics}, we validate the robustness of our method  numerically.

\section{Problem Formulation and Theoretical Analysis}

\subsection{Problem Formulation}
The governing equation for propagation of heat is the heat equation. Its most rudamentary form is as follows
\begin{align}
    \begin{cases}\partial_t u(t,x) - \Delta u(t,x) &= 0,  \ t \in (0,T),  \ x\in \R^2 \\
    u(0,x)&=u_0, \ x \in \R^2.
    \end{cases} \label{eq:heat}
\end{align}
Here $\Delta = \nabla^2 = \partial_{x_1}^2 + \partial_{x_2}^2$ is the Laplace operator and $u_0$ is the initial heat distribution. In this work, we will model $u_0$ it as a linear combination of finitely many point sources:
\begin{align}
    u_0 = \mu_0 = \sum_{i=1}^s c_i \delta_{p_i}. \label{eq:groundTruth}
\end{align}
where $s \in \N$, $c_i>0$ and $p_i \in \R^2$ denote the number of sources, their amplitudes and locations, respectively. Note that although this model is quite idealized, it is commonly used in the literature (often under the name of \emph{instantaneous sources}). We choose to denote the ground truth distribution  by $\mu_0$, since it is a measure. We will for convenience make the global normalization assumption
\begin{align*}
    \sum_{i=1}^s c_i =1.
\end{align*}
It is well known that \eqref{eq:heat} has a unique solution for initial values of the type \eqref{eq:groundTruth} -- and even that the solution is given by convolution of $u_0$ with the Green function $G(x,t) = (4\pi t)^{-1} e^{-\frac{1}{2t}\abs{x}^2}$, i.e.
\begin{align}
    u(x,t) = \int_{\R^2} G(p-x,t) d\mu_0(p) \label{eq:HeatSol}
\end{align}

 As was already mentioned, $\mu_0$ is a measure. More specifically, $\mu_0$ is a member of the space $\calM(\R^2)$ of \emph{(signed) Radon measures of finite variation}. That space is naturally equipped with the \emph{total variation norm}:
\begin{align*}
    \norm{\mu}_{TV} = \sup_{ \substack{\bigcup_{i=1}^N U_i = \R^2\\ U_i \text{ disjoint.} }} \sum_{i=1}^N \abs{\mu(U_i)}
\end{align*}
The $TV$-norm can intuitively be viewed as the infinite-dimensional analogue of the $\ell_1$-norm of a vector in $\R^n$. This becomes especially clear when considering that the $TV$ norm of a train of $\delta$-peaks \eqref{eq:groundTruth} is equal to $\sum_{i=1}^s \abs{c_i}$. Knowing that $\ell_1$-minimization promotes sparsity, it seems reasonable to minimize the $TV$-norm to recover a signal like \eqref{eq:groundTruth} from linear measurements $b=M\mu$:
\begin{align}
    \min \norm{\mu}_{TV}~\st~M\mu =b \tag{$\calP_{TV}$}.
\end{align}
$b$ is thereby given through samples of the function \eqref{eq:HeatSol}, expressed by the linear measurement operator $M$. This idea is per se not new (see for instance \cite{tang2013compressed,candes2014towards,DuvalPeyre2015}), but we believe that it is the first time that this method is proposed for solving the source localization problem.

\subsection{Theoretical Analysis}
Very generally, the solution of a problem of the type $(\calP_{TV})$ (and also of the type $\calP_{TV}^{\rho,e}$ defined below) has a structure of the form
\begin{align*}
    \mu^* = \sum_{i=1}^m d_i \delta_{p_i^*},
\end{align*}
where $m$ is the dimension of the measurement vector $b$ (see \cite{flinth2017exact,unser2016splines}). In this paper, we will use the theory from \cite[Sec. 4.3]{flinth2017soft} to theoretically guarantee that the positions of the sources in the reconstructed signal $p_i^*$ are at least close to the positions of the ground truth sources.   Thereby, we in some sense circumnavigate the fact that the dictionary (i.e. the operator $M$) is highly coherent. In fact, we will even extend the soft recovery framework to also include noisy measurements.

We start by quickly reviewing the mentioned theory (we leave out a lot of technical details for now -- they are instead discussed in the appendix \ref{app:soft} ). Let $\phi \in L^2(\R^2)$ be an $L^2$-normalized low-pass filter: To be concrete, let us say $\phi(x) = (2\pi \Lambda)^{-1/2} \exp(- \abs{x}^2/(2\Lambda))$ for some $\Lambda>0$. Define the Hilbert space $\calE$ as the space of tempered distributions having the property $v*\phi \in L^2(\R^2)$, with scalar product
\begin{align}
    \sprod{v,w}_\calE = \sprod{v*\phi, w*\phi}_{L^2(\R^2)} = \int_{\R^2} \hat{v}(\xi) \overline{\widehat{w}}(\xi) \abs{\widehat{\phi}(\xi)}^2 d\xi.
\end{align}
$(\delta_p)_{p \in \R^2}$ is then a normalized dictionary in $\calE$, since   $ \norm{\delta_p}_\calE^2 =  \sprod{\delta_p, \delta_p }_\calE = \int_{\R^2} \exp(i \xi \cdot p) \exp(-i \xi \cdot p) \abs{\widehat{\phi}(\xi)}^2 d\xi = 1$. In particular,  our ground truth measure \eqref{eq:groundTruth} is a member of $\calE$.
Let us furthermore associate an \emph{autocorrelation function} $a$ to the filter $\phi$ through
\begin{align*}
    a(x) &= \phi * \phi (x) = \int_{\R^2} \abs{\widehat{\phi}(\xi)}^2 \exp(-i x \cdot \xi) d\xi \\
    &= \exp\left(- \frac{\abs{x}^2}{4\Lambda}\right). 
\end{align*}
We are now ready to cite the result we wish to apply.
 \begin{theo} \cite[Cor. 4.3]{flinth2017soft} \label{theo:softRec} Let $\calE$, $\mu_0$, $c_i$, $p_i$ and $a$ be as above, $M: \calE \to \R^d$ (for some $ d\in \N$) be a continuous linear operator, and $i_0 \in \set{1, \dots, s}$ be arbitrary. If there exists a $g \in \ran \calF^{-1} \abs{\widehat{\phi}}^2 \calF M^*$, where $\calF$ denotes the Fourier transform, with
 \begin{align}
     \sum_{i =1}^s  \re(c_i^0 g(p_i)) \geq 1, \quad \abs{g(p_{i_0})} \leq \sigma, \nonumber \\
     \text{ and } \sup_{p \in \R} \abs{ g(p) - a(p-p_{i_0}) g(p_{i_0})} \leq 1-\tau, \label{eq:gCond}
 \end{align}
 where $\sigma\geq 0, \tau \in (0,1]$ are parameters, then for every solution $\mu^*$ of $\calP_{TV}$, there exists an $p^* \in \supp \mu^*$ with
 \begin{align}
     \abs{a(p^*-p_{i_0})} \geq \frac{\tau}{\sigma}. \label{eq:errorBound}
 \end{align}
    \end{theo}
    
    \begin{rem}\label{rem:pointsClose}
    Note that for our function $a$, a bound of the form \eqref{eq:errorBound} immediately implies a bound on the proximity of $p^*$ to $p_{i_0}$:
    \begin{align*}
        \abs{p^*-p_{i_0}} \leq \sqrt{4\Lambda \log \left(\tfrac{\tau}{\sigma}\right)}.
    \end{align*}
     
    \end{rem}
    
    In practice, measurements are always contaminated with noise -- that is, we do not have access to $b= M\mu$, but rather a noisy version $b = M \mu + e$, with $\norm{e}_2 \leq \epsilon$. A canonical extension of $(\calP_{TV})$ to handle this setting is to consider the following ''$TV-LASSO$'' (the term $LASSO$  is from \cite{tibshirani1996regression}, where it was coined in a finite-dimensional setting):
    \begin{align}
        \min \norm{M\mu - b}_2~\st~\norm{\mu}_{TV} \leq \rho \tag{$\calP_{TV}^{\rho,e}$},
    \end{align}
    where $\rho>0$.  In the appendix, we will extend the soft recovery framework to include also such problems. For the $TV$--problem, we obtain the following result.
    \begin{cor} \label{cor:gStable} Under the assumptions of \ref{theo:softRec}  
  there exists for every solution $\mu^*$ of $(\calP_{TV})$ an $p \in \supp \mu^*$ with
 \begin{align*}
 	\abs{a({p-p_{i_0}})} \geq \frac{\tau}{\sigma} - \frac{ 2\norm{\lambda}_2 \epsilon+(\rho-1)}{\rho\sigma},
 \end{align*}
 where $\lambda$ is defined through $g = \calF^{-1} \abs{\widehat{\phi}}^2 \calF M^* \lambda$, where $g$ is as in \ref{theo:softRec}.
 \end{cor}

\begin{rem} 
The previous theorems only secure that each peak in the ground truth measure can be approximately retrieved in the solution $\mu_*$ (i.e. \emph{softly recovered}). This is of course important on its own, but let us also note that an approximate recovery of the peaks is enough to secure good  reconstruction of the \emph{temperature field} $\int_{\R^2} G(p-x) d\mu_0(x)$! If $v$ and $v'$ are close, the corresponding active dictionary elements $G(\cdot -v)$ and $G(\cdot -v')$ also are, so interchanging them in the decomposition does not have much of an effect.
\end{rem} 
    
   In \cite{flinth2017soft}, the case that the measurement operator $M$ consists of sampling the coefficients of the signal $\mu_0$ in a so-called frame of $\calE$ is discussed. A \emph{frame} is thereby a family $(f_i)_{i\in I}$ of vectors in a Hilbert space $\calH$ obeying the \emph{frame inequality}
   \begin{align*}
       \forall v \in \calH, \alpha \norm{v}_\calH^2 \leq \sum_{i \in I} \abs{\sprod{v,f_i}}^2 \leq \beta \norm{v}_\calH^2
   \end{align*}
   for some scalars (\emph{frame bounds}) $\alpha, \beta >0$.
   
   In this work, our measurements are of an entirely different nature: namely  samples of the solution, \eqref{eq:HeatSol}, of the heat equation. If we call the sample set $\calS$, the measurement operator is hence given by
   \begin{align} \label{eq:Measurements}
        M \mu  = \left(\int_{\R^2} G(p-x,t)d\mu(p) \right)_{(x,t)\in \calS}.
   \end{align}
   Note that $M$ is defined on the whole of $\calE$, since we can convolve the Green function with any tempered distribution, in particular elements of $\calE$. $M$ is, provided $t > \Lambda$, furthermore continuous as an operator on $\calE$, which the following lemma shows.  Its proof can be found in Section \ref{sec:Mcont} of the Appendix.
   \begin{lem} \label{lem:MProps}
    If $t>\Lambda$ for all $(x,t) \in \calS$, $M$ is a continuous operator from $\calE$ to $\C^{d}$, with $d= \abs{\calS}$. Its adjoint is given by
    \begin{align*}
        (M^*\lambda)(p) = \sum_{(x,t) \in \calS} \lambda_{x,t} \widetilde{G}(p-x,t) ,
    \end{align*}
    where $\widetilde{G}$ is defined as $\calF^{-1} \abs{\hatphi}^{-2} \calF G$.
   \end{lem}
   
   Lemma \ref{lem:MProps}, Theorem \ref{theo:softRec} and Remark \ref{rem:pointsClose} tell us that if we can construct a function $g$ of the form
   \begin{align} \label{eq:gForm}
        g(p) &= \sum_{(x,t) \in \calS} \lambda_{x,t} G(p-x,t)  \nonumber\\
        &= \sum_{(x,t) \in \calS} \lambda_{x,t} (4\pi t)^{-1} \exp(- \abs{p-x}/(2t))
   \end{align}
which satisfies \eqref{eq:gCond}, we are done (note that $\calF^{-1} \abs{\hatphi}^2 \calF \widetilde{G}=G$). The rest of the section will be devoted to this cause. 

The strategy will be the following: If we could choose $g$ as $(c_{i_0}^0)^{-1}a(\cdot-p_{i_0})$, we would immediately have \eqref{eq:gCond} for $\tau=1$ and $\sigma = (c_{i_0}^0)^{-1}$ (remember that we assumed that $c_i>0$ for all $i$, and that $a$ is a positive function). $a(p-p_{i_0}) = \exp(-\abs{p-p_{i_0}}^2/(4\Lambda))$ is however not always of the form \eqref{eq:gForm} -- the Gaussians in  \eqref{eq:gForm} neither have the same width nor the same centers as the Gaussian $a(p-p_{i_0})$. The former problem can be solved by assuming $2t= 4\Lambda$ for  $(x,t) \in \calS$ -- note that $\Lambda$ can per se be freely chosen, so this equality is always achievable. The latter problem is then to approximate a Gaussian of width $\sqrt{4\Lambda}$ centered at $p_{i_0}$ with other Gaussians of width $\sqrt{4\Lambda}$ centered at the points $x$ on the grid. In order to be able to give a concrete proof that this is possible, let us make the following assumptions:
\begin{itemize}
\item $p_{i_0}$ is located in the rectangle $[-1/2, 1/2]^2$.
\item The samples are all taken at time $t =2 \Lambda$ and spatially on a uniform grid over $[-1,1]^2$ with spacing $\tfrac{1}{m}$ for some $m \in \N$, i.e.
\begin{align}
    x_{n_1, n_2} = (\tfrac{n_1}{m}, \tfrac{n_2}{m}), \quad n \in \set{-m, \dots m}^2. \label{eq:grid}
\end{align}
\end{itemize}

Given this structure of the samples, the mathematical task at hand is to do the following: Denoting our Gaussian with $\psi$,  we need to approximate $\psi(\cdot-p_{i_0})$ with a linear combination of the functions $\psi(\cdot - \tfrac{n}{m})$, $ n \in  \set{-m, \dots, m}^2$. By translating everything and possibly discarding some of $n$'s, we can bring this down to approximating $\psi( \cdot - \tfrac{\Delta}{m})$ with functions $\psi( \cdot - \tfrac{n}{m})$, $\abs{\Delta}< \tfrac{1}{2}$ and $\abs{n_1}, \abs{n_2} \leq \tfrac{m}{2}$:
\begin{align} \label{eq:appr}
    \psi\left(x- \tfrac{\Delta}{m}\right) \approx \sum_{n\in \set{-\tfrac{m}{2}, \dots, \tfrac{m}{2}}^2} c_n\psi\left(x- \tfrac{n}{m}\right).
\end{align}
The  approximate equality is thereby supposed to be true uniformly in $x \in \R^2$. Since
\begin{align} \label{eq:FourierTrick}
\abss{\psi\left(x- \tfrac{\Delta}{m}\right) - \sum_{n\in \set{-\tfrac{m}{2}, \dots, \tfrac{m}{2}}^2} c_n \psi\left(x- \tfrac{n}{m}\right)} 
 &= \abss{ \int_{\R^2} \widehat{\psi}(\xi) \left( \exp(i \tfrac{1}{m} \Delta \cdot \xi) -\sum_{n} c_n\exp(i \tfrac{1}{m}n \cdot \xi)\right) d\xi}  \\
& = m^2 \abss{ \int_{\R^2} \widehat{\psi}(m\omega) \left( \exp(i  \Delta \cdot \omega) -\sum_{n} c_n\exp(i n \cdot \omega)\right) d\omega}, \nonumber
\end{align}
the problem boils down to the question: \emph{How well can we approximate the function $\omega \mapsto \exp(i \Delta \cdot \omega)$ with sums of complex exponentials $\omega \mapsto \exp(i n \cdot \omega)$ with $n \in \Z^2, \norm{n}_\infty \leq \tfrac{m}{2}$?} (Note that we from the second line and forward left out the range of summation for $n$ in order to not overload the notation. We will continue doing this in the sequel.)


\begin{figure}[t]
\begin{center}
     \includegraphics[width=90mm]{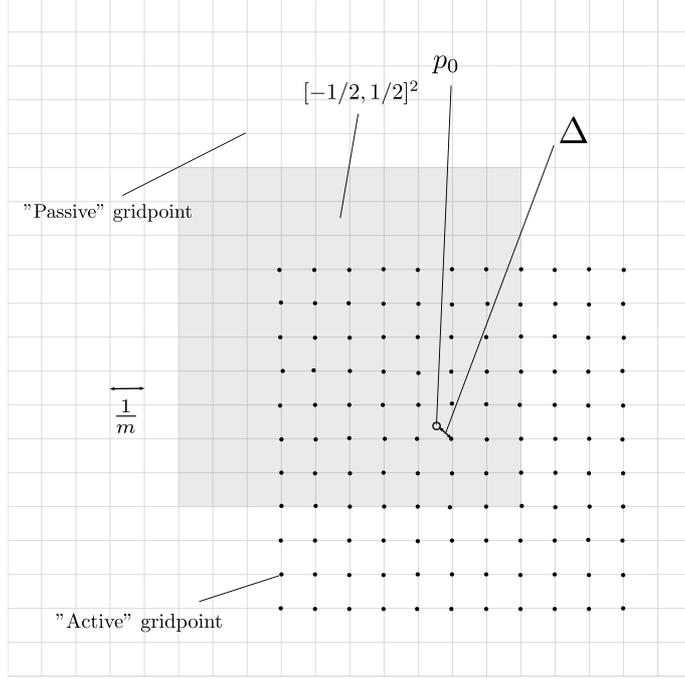}   
\end{center}
    \caption{ For all points $p_{i_0} \in [-1/2,1/2]$, we are sure that at least the functions $\psi$ with centers  in the points $p_{i_0} + n/m$, $\norm{n}_\infty \leq \tfrac{m}{2}$ are at our disposal to approximate $\psi(\cdot - p_{i_0})$. }
\end{figure}

This question can be solved with routine techniques -- a detailed proof is given in the appendix. One arrives at the following statement.
\begin{lem} \label{cor:scalars}
    There exists scalars $c_n$ with
    \begin{align*}
        \sup_{p \in \R^2} \bigg\vert  a(p-p_{i_0}) - \sum_{n }c_n (2\pi t)^{-1} \exp(- \abs{p-\tfrac{n}{m}}/(2t)) \bigg\vert 
        \leqsim m^{-1}(1+ \Lambda^{-1/2})
\end{align*}
The vector $c$ obeys $\norm{c}\leq 1$.
\end{lem}

A full proof of Lemma \ref{cor:scalars} can be found in Section \ref{sec:Appr} of the appendix. Here, we will instead state and prove the main theoretical result of this paper.

\begin{theo}[Main Result] \label{th:main}
    Let $\Lambda>0$ and suppose that the sample set $\calS$ has the structure \eqref{eq:grid} with $m \geqsim (1+\Lambda^{-1/2})(c_{i_0}^0)^{-1}$. Then if $b=M\mu_0,$ then for every minimizer $\mu_*$ of $\calP_{TV}$, there exists a $p_* \in \supp \mu_*$ with $\abs{p_{0}-p^*}\leq \sqrt{4\Lambda \log\left(2/c_{i_0}^0\right)}$.
    
    In fact, also in the case that $b=M\mu_0 + e$ with $\norm{e}_2 \leq \epsilon$, the regularized problem $(\calP_{TV}^{\rho,e})$ for every $\rho \geq 1$ has the following property: For every minimizer $\mu_*$ of $\calP_{TV}^{\rho,e}$, there exists a $p_* \in \supp \mu_*$ with $\abs{p^*-p_{i_0}} \leq\sqrt{ 4\Lambda \log\left(\left(\frac{c_{i_0}^0}{2} - \frac{6(2\epsilon+ (\rho-1))}{8\rho}\right)^{-1}\right)}$ . Put shortly, the bound for the noiseless case deteriorates gracefully with a non-optimal choice of $\rho$ and increasing noise level.
\end{theo}
\begin{proof} Lemma \ref{cor:scalars} together with the structure of $\calS$ and the form  of $\calF^{-1} \abs{\hatphi}^{2} \calF M^*$ (see Lemma \ref{lem:MProps}) implies that for any $\hatsigma$, there exists a $g \in \ran \calF^{-1} \abs{\hatphi}^{2} \calF M^*$ with
    \begin{align*}
        \sup_{p \in \R^2} \abs{g(p) - \hatsigma a(p-p_{i_0})} \leqsim \hatsigma m^{-1}(1+\Lambda^{-1/2}).
        \end{align*}
        This $g$ satisfies, for a constant $K>0$,
        \begin{align*}
            \sum_{i =1}^s  \re(c_i^0 g(p_i))& \geq \sum_{i =1}^s  \re(c_i^0 \hatsigma a(p_i-p_{i_0})) - \sum_{i=1}^s c_i^0 \abs{g(p_i)-\hatsigma a(p_i-p_{i_0})} \\
            &   \geq \re(c_{i_0}^0 \hatsigma a(p_{i_0}-p_{i_0})) - K\sum_{i=1}^s c_i^0 \hatsigma m^{-1}(1+\Lambda^{-1/2})   \geq \hatsigma (c_{i_0}^0 - Km^{-1}(1+\Lambda^{-1/2}) ), 
            \end{align*}
            where we used that both the coefficients $c_i^0$ as well as the function $a$ are positive when discarding all terms but the one related to $i_0$ in the second step. In the last step, we inferred the normalization assumption on $c^0$. We furthermore have
            \begin{align*}
            \abs{g(p_{i_0})} &\leq \hatsigma a(0) + \hatsigma \abs{\hatsigma a(p_{i_0}-p_{i_0})- g(p_{i_0})}  \leq \hatsigma\left( 1+Km^{-1}(1+\Lambda^{-1/2})\right) \\
            \sup_{p \in \R} \abs{g(p) - a(p-p_{i_0})g(p_{i_0})}   &\leq \sup_{p \in \R} \abs{g(p)-\hatsigma a(p-p_{i_0})}+ \abs{\hatsigma - g(p_{i_0})}\sup_{p \in \R} \abs{a(p-p_{i_0})} 
           \leq 2\hatsigma Km^{-1}(1+\Lambda^{-1/2}).
        \end{align*}
        Hence, if we choose $\hatsigma := (c_{i_0}^0 - Km^{-1}(1+\Lambda^{-1/2}) ) ^{-1}$, \eqref{eq:gCond} is satisfied for 
        \begin{align*}
            \sigma &= \frac{1+ Km^{-1}(1+\Lambda^{-1/2})}{c_{i_0}^{0} - Km^{-1}(1+\Lambda^{-1/2})}, \\
             \tau&= 1- 2K\frac{m^{-1}(1+\Lambda^{-1/2})}{c_{i_0}^{0} - Km^{-1}(1+\Lambda^{-1/2})} = \frac{c_{i_0}^0-3Km^{-1}(1+\Lambda^{-1/2})}{c_{i_0}^{0} - Km^{-1}(1+\Lambda^{-1/2})}.
        \end{align*}
        Now Theorem \ref{theo:softRec} implies that any minimizer of $\calP_{TV}$ possesses a point $p^*$ in its support with
        \begin{align*}
            \exp(-\abs{p^*-p_{i_0}}^2/(4\Lambda)) \geq \frac{\tau}{\sigma} d= \frac{c_{i_0}^0-3Km^{-1}(1+\Lambda^{-1/2})}{1+Km^{-1}(1+\Lambda^{-1/2})} 
            \geq \frac{c_{i_0}^0}{2}
        \end{align*}
if we choose $m \geq 7K(1+\Lambda^{-1/2}) (c_{i_0}^0)^{-1}$. This implies $\abs{p^*-p_{i_0}} \leq \sqrt{4\Lambda \log\left(2/c_{i_0}^0\right)}$. 

As for the noisy case, Corollary \ref{cor:gStable} implies that any minimizer of $\calP_{TV}^{\rho,e}$ possesses a point $p^*$ in its support with
\begin{align*}
            \exp(\abs{p^*-p_{i_0}}^2/\Lambda) \geq \frac{\tau}{\sigma} - \frac{2\epsilon+ (\rho-1)}{\rho\sigma}
        \end{align*}
        We used that $\norm{\lambda}_2 \leq 1$. Now, the first one of these terms were estimated to be larger than $c_{i_0}^0/2$ in the argument for the non-noizy case. As for the second, we have:
        \begin{align*}
            \sigma = \frac{1+ Km^{-1}(1+\Lambda^{-1/2})}{c_{i_0}^{0} - Km^{-1}(1+\Lambda^{-1/2})} \geq \frac{1+ 7^{-1} c_{i_0}^0}{c_{i_0}^{0} - 7^{-1}c_{i_0}^0} \geq \frac{8}{6},
        \end{align*}
since $c_{i_0}\leq 1$. This, together with the estimate on $\tau/\sigma$, implies
\begin{align*}
    \frac{\tau}{\sigma} - \frac{2\epsilon+ (\rho-1)}{\rho\sigma} \geq \frac{c_{i_0}^0}{2} - \frac{6(2\epsilon+ (\rho-1))}{8\rho},
\end{align*}
        which yields $$\abs{p^*-p_{i_0}} \leq \sqrt{ 4\Lambda \log\left(\left(\frac{c_{i_0}^0}{2} - \frac{6(2\epsilon+ (\rho-1))}{8\rho}\right)^{-1}\right)}.$$
 \end{proof}
 
 \begin{rem}
      As is common in the compressed sensing literature, the theorem provides a lower bound on the amount of measurements $d=m^2$ needed to secure approximate recovery of the source positions. The bound grows with decreasing $\Lambda$ and amplitude $c_{i_0}$. This is sound: The lower the $\Lambda$, the higher is the precision, and the smaller $c_{i_0}^0$, the less significant is the peak.
      
      Note that if all peaks are equally large, $c_{i_0}^{-1}$ exactly equals the sparsity $s$ of the signal (remember that we assumed $\sum c_i =1$.) $m \geqsim c_{i_0}^{-1}$ hence corresponds to $d\geqsim s^2$, which unfortunately is suboptimal.
      (A linear dependence on the sparsity is asymptotically optimal -- even if we are a priori given the positions $p_i$, we will still need $s$ measurements to determine the amplitudes $(c_i)_{i=1}^{s}$.) This may well be an artefact of the proof, since it solely relies on approximation properties of the Fourier basis (which is well known to struggle in high dimensions). We leave the question whether an improvement is possible as an open problem. 
      
      Reducing the number $d$ practically corresponds to lowering the number of spatiotemporal samples. This will reduce the financial cost.  Also, it will enable satisfactory reconstruction results in applications where there are very tight constraints on the number of sensors or time samples. 
 \end{rem}
 
 This concludes our theoretical analysis, and we move on to test our method numerically.
 
\section{Numerical Experiments} \label{sec:Numerics}

In this section, we numerically demonstrate that our proposed method performs well in practice. To start, we work in a one-dimensional regime. The main reason for this is to be able to make a comparison with the method proposed in \cite{ranieri2011sampling}, which to the best knowledge of the authors is the state-of-the art compressed-sensing based method for recovery of sparse \emph{initial} source distribution in thermal source localization. Note that there are more recent methods for thermal field reconstruction, as was explored in the introduction. These can however not be fairly compared to our method. The method the authors of  \cite{kitic2016physics,kitic2017versatile} mainly advertise uses a different type of sparsity than we do. In their setting, $Au$ is sparse in a certain sense, where $A$ is an \emph{analysis} operator. This is completely different from our sparsity, which is in the direct sense for the initial condition.  On the other hand, the works \cite{murray2015estimating,murray2017universal} assume a model in which the sparse recovery problem can be reformulated as the problem of recovering a sum of sinusoids with unknown frequencies, where the Prony method is applicable. This is not immediate in our setting.

Although the theory developed previously applies to the 2D-case, we will first stay in a 1D-regime here. The reason for this is twofold: First, \cite{ranieri2011sampling} only consider the 1D-case, so only in this regime, a fair comparison can be made. By
doing extensive numerical experiments on synthetic data, we compare the performance of the method from \cite{ranieri2011sampling} and the one proposed in this work, in several cases. In particular, we consider scenarios with noisy measurements, sources off-the-grid, and also cases where the number of samples is too small  for the method \cite{ranieri2011sampling} to work.

In the second part of this section, we will present the results also for 2D case. Let us begin by describing the method in \cite{ranieri2011sampling}. 

\subsection{Brief Review of the Structure of the Simulation and Comparing Method} \label{NUM:SecI}

The authors in \cite{ranieri2011sampling} introduced a sampling procedure for fully reconstructing the unknown initial field distribution from spatiotemporal
samples collected at time $t>0$ from a diffusion field which is generated by $s$ sources. By constructing a sensing matrix through discretizing the diffusion field, and reformulating this matrix as a function of spatial and temporal sampling densities, they derive precise bounds for spatial and temporal densities under which the sensing matrix is well-conditioned for solving this inverse problem. 

For introducing the initial inverse problem in discrete form similar to \cite{ranieri2011sampling}, we assume that the $s$ sources are deployed on a grid of size $P$, with some separation.  Then, we define $X_{P}$ for representing the spatial location of these sources as follows:

\begin{equation} 
X_{P}=\{m\varDelta_{1}\,:\,0\leq m\leq P-1\}, \label{eq:5}
\end{equation}

where $\varDelta_{1}=\frac{2 \pi}{P}$. We hence assume that the diffusion
field propagates in one dimension with length $2 \pi$. Also, we define a vector $\mu_0 \in \R^p$
 representing the amplitude of the source. Concretely, the amplitude at $m$-th position of the
grid $X_{P}$  is given by

\begin{equation*}
\mu_0(m)=c_{m}.\label{eq:8}
\end{equation*}

We will assume that $\mu_0$ is $s$-sparse, with $s \ll M$, similar to \cite{ranieri2011sampling}. Note that this is a slight simplification compared to \eqref{eq:groundTruth}, where the sources are not confined to lie on the grid.

For sensing the diffusive field in this scenario, we consider a sensor
network with $N_{s}$ spatial sensors deployed uniformly in $[0,2 \pi]$, each one collecting $N_{t}$ uniform time samples. Hence, we
have a spatiotemporal sampling procedure which provides us $N_{s}\times N_{t}$ samples and we aim to estimate the initial source parameters, $x_{m}$
and $c_{m}$, using these samples.

Similar to \cite{ranieri2011sampling}, we also define
$Y_{N_{s}}$ and $\mathit{T}_{N_{t}}$ for representing the spatial
position of sensors and temporal sampling grid, respectively:

\begin{equation*}
Y_{N_{s}}=\{n\varDelta_{2}\,:\,0\leq n\leq N_{s}-1\},\label{eq:9}
\end{equation*}

\begin{equation*}
\mathit{T}_{N_{t}}=\{\ell\tau\,:\,0\leq\ell\leq N_{t}-1\},\label{eq:10}
\end{equation*}

where $\varDelta_{2}$ is the distance between spatial sensors in spatial domain. $\tau$ denotes the time interval between two time samples. Then, the vector $b \in \R^d$ is obtained by concatenating all the $d=N_{s}N_{t}$ samples collected by all the sensors, aiming to find a sensing matrix $M$ with size $d\times P$ for relating $b$ and $\mu_0$ as a discrete linear equation as follows:

\begin{equation}
b=M \mu_0 \label{eq:10-1}.
\end{equation}

The matrix $M$ is constructed by using a  discrete version
of the Green function, assuming the diffusion parameter $\gamma=1$ without loss of generality, and the position of sensors to be $n\Delta_{2}$, $n=1, \dots, N_s$:

\begin{equation}
M_{m}^{\ell}(n):=\frac{1}{\sqrt{4\pi\ell\tau}}\exp\{-\frac{(n\Delta_{2}-m\Delta_{1})}{4\ell\tau}\}\quad for\,\ell=1,...,N_{t}
\label{eq:12}
\end{equation}

Equation \eqref{eq:12} describes the elements of the sub-matrix $M^{\mathrm{\ell}}$ which results when sampling the field by collecting all spatial samples at one time
instant, $t=\ell\tau$. The following  equation is then a discretized version of the inverse problem considered in this publication:

\begin{equation}
 b^{\mathrm{\ell}}=M^{\mathrm{\ell}}\mu_0\label{eq:13}
\end{equation}

where $b^{\mathrm{\ell}}$ is the vector collected all spatial
samples at the sampling instant $t=\ell\tau$. Finally, the matrix $M$ is constructed by concatenating $N_{t}$ sub-matrices $M^{\mathrm{\ell}}$
into one matrix. Accordingly we define a vector $b \in \R^d$ by concatenating the
$N_{t}$ vectors $b^{\mathrm{\ell}}$, each one containing
all spatial samples at each sampling time \cite{ranieri2011sampling}.

Finally, by introducing a parameter $\rho:=\frac{\tau}{(\Delta_{2})^{2}}$, \cite{ranieri2011sampling} proposed the  following bounds on $\rho$:

\begin{equation}
\frac{1}{2N_{t}}=\rho_{MIN}<\rho<\rho_{MAX}=\frac{(N_{s}-1)^{2}}{72\times N_{t}}.\label{eq:16}
\end{equation}
These bounds imply that the system model matrix suitable for solving the target inverse problem. \cite{ranieri2011sampling}. 

\subsection{Description of our Method} \label{NUM:SecII}

Now we describe how numerically implement our proposed method.  A diffusive field, generated by $s$ initial sources is considered. As mentioned in the previous section, we do not assume a fixed grid for deploying the sources. Therefore, the measurements are directly captured using the closed form solution of  the Green function. Using Equation \eqref{eq:HeatSol}, we see that the following equation can be used to collect the samples: 

\begin{align} \label{eq:GreenFuncSol}
u(x,t)=\sum_{i=1}^{s}  \frac{c_i}{\sqrt{4 \pi t}} \exp\left\{-\frac{(x - p_i)^2}{2 t}\right\} , \quad (x,t) \in \calS
\end{align}

where the $N_s \times N_t$ samples are taken by evaluating \eqref{eq:GreenFuncSol} on some set of sampling points $\calS$. We can, and will, choose $\calS= Y_{N_s} \times T_{N_t}$, but other choices are certainly possible.

The values of the source positions, $p_i$, are also assigned in a grid-less manner based on on/off-grid settings which we will consider in our simulation scenario. The captured spatiotemporal samples are then concatenated to form the vector $b$. For demonstrating the robustness of algorithm against insufficiently many samples, as well as staying as close as possible to the theoretical setting, we sense the field at only one time instant. Therefore, we only access to one time sample for the whole sensing procedure. 

After obtaining the measurement vector $b$, our goal is to solve either $(\calP_{TV})$, or $(\calP_{TV}^{\rho,e})$ (depending on whether the measurements are noisy or not). This is in general hard. There are a few special cases in which the problem $(\calP_{TV})$ can be reduced to a finite-dimensional problem -- most notably in the case of Fourier measurements \cite{candes2014towards,dossal2016sampling}. (See also \cite{flinth2017exact} for a few more examples.)
In this setting, however, such a reduction seems very hard, whence we instead propose the following (heuristic) scheme\footnote{Developing a theory for such a discretization is an interesting line of research for future work. Some initial results concerning discretizations of this kind can be found in \cite{DuvalPeyre2015}.}:  We initially restrict our analysis to a rough grid, (as in \eqref{eq:5}, but with $P$ small). We then solve the dual problem of the coursely discretized problem. 
  For the infinite-dimensional dual problem, the points where the \emph{dual certificate} $\nu_{\infty}=M^*p_{\infty}$ has modulus one exactly corresponds to the positions of the peaks in the solution $\mu_0$ of $(\calP_{TV})$ \cite{DuvalPeyre2015}. Thus, by observing where the value of the dual certificate $\nu_* = M^*p_*$ (where $p_*$ is the solution of $(\calP_{Dual})$ is larger than a threshold $\tau_{peak}$ close to $1$ (Figure \ref{Dual_Solution_Selected}), we get a rough idea where the peaks of the infinite-dimensional solution $\mu_{\infty}$ are located.
  
  For the non-regularized problem $(\calP_{TV})$, the dual of the discretized problem has a simple structure:
  \begin{align}
    \max_p~Re\{<b,p>\}~\st~\norm{M_{\text{discrete}}p}_{\infty}\leq 1. \tag{$\calP_{Dual,TV}$}
\end{align}
This is no longer the case for the problem $\calP_{TV}^{\rho,e}$ \cite{osborne2000lasso}. Therefore, we deviate a bit from the theory presented above and instead, as is usual, consider the following, unconstrained version of the LASSO:
\begin{align}
    \mu^{*}= \min_{\mu} \tfrac{1}{2}\norm{M_{discrete}\mu -b}_{2}^2 + \lambda \norm{\mu}_{1} \tag{$\calP_{LASSO}$}.
\end{align}
It is well known that for each parameter $\rho$, there is a $\lambda$ such that the solution of $(\calP_{TV}^{\rho,e})$ is equal to the one of $(\calP_{LASSO})$ \cite[Theorem B.28, p.562]{MathIntroToCS}. Considering the fact that the parameter $\rho$ (or $\lambda $, respectively) anyhow needs to be fine-tuned, making this transition is justified. The dual problem of $(\calP_{LASSO})$ again has a simple form:
\begin{align}
     \min_p~\norm{\tfrac{b}{\lambda}-p}_2~\st~\norm{M_{\text{discrete}}p}_{\infty}\leq 1 \tag{$\calP_{Dual,LASSO}$}.
\end{align} 

\begin{algorithm} 
		\caption{ Summary of the method we propose.} \label{alg:OMP}
		\KwData{ A measurement operator $M: \calM \to \R^d$, and (noisy) measurements $b \in \R^d$.}
		\KwResult{ An estimate $\mu_*$ of a sparse approximate solution to $M\mu=b$.}
		
	\nl	Initialize a course grid $X$.
		
	\Repeat{Stopping condition satisfied}{
		  \nl Find a solution $p$ to the dual problem of $(\calP_{TV})$ or $\calP_{LASSO}$, respectively( depending on whether $b$ is contaminated with noise or not ), discretized to the grid $T$
		   
		  \nl Find points $(q_i)_{i=1}^r$ in which $\nu = M_{\text{discrete}}^*p$ has large absolute value.
		  
		  \nl Refine $X$ by adding points close to $(q_i)_{i=r}$
		  }
		  
		  \nl Use final dual certificate to define final grid $X_*$
		  
		  \nl Output $\mu_* = (M\vert_{{X_*}})^\dagger b$
		
	\end{algorithm}


In next step, we refine the resolution of our dictionary by adding extra points close to the selected points from the previous step, re-discretize, and solve the corresponding dual problem (see Figure \ref{Dual_Solution_Selected}). We repeat the same procedure until a stopping criterion is met.

\begin{figure}[t]
\begin{center}
     \includegraphics[width=90mm]{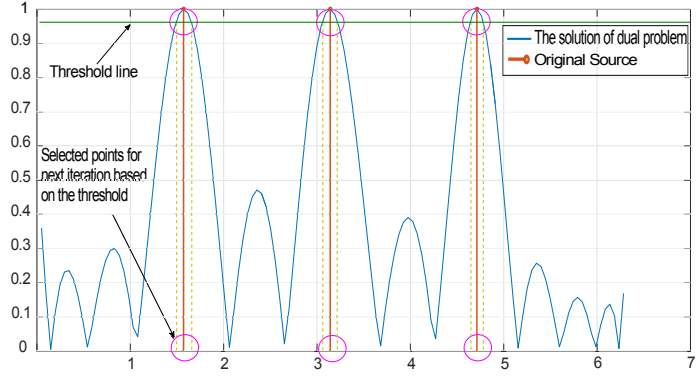}   
\end{center}
    \caption{\label{Dual_Solution_Selected} As can be seen, the points are selected based on the threshold which put on the peak solutions of the dual problem. The system model dictionary matrix is refined for the next iteration by adding extra atoms corresponding to these points}
\end{figure}

The final dual certificate $\nu_{final}$ is then used to obtain the estimated position of the sources. Concretely, the points $x$ with $\abs{\nu_{final}(x)}\approx 1$ are grouped into clusters (this is easy, since we are in a $1D$-regime). For each such cluster, we choose the midpoint $x_i$ to be a member on the final grid $X_*$.

In the final step, we build the solution $\mu_*$ by making an ansatz in the form of a sum of Dirac-$\delta$'s supported on $X_*$. The amplitudes of the peaks are obtained by solving the inverse problem $b=M \mu_0$ using the pseudo inverse of the dictionary matrix constructed by the atoms corresponding to the final selected points $X_*$: 
\begin{align*}
\mu_* = \left(M\vert_{X_*}\right)^{\dagger} b,
\end{align*}
where $M\vert_{X_*} c = \sum_{x\in X_*} c_x M(\delta_x)$. Under the assumption that the set of chosen points $X_*$ contains for each ground truth peak $p$ exactly one point close to $p$, and apart from that only points far away from all points in the ground truth support $X_0$, $\mu_*$ will be close to the ground truth signal. The following formal statement holds.
\begin{prop} \label{prop:MoorePenrose}
    Let $X_0=\set{x_1, \dots, x_s}$ be a set of ground truth peaks with the property that the matrix $M_{\vert_{X_0}}$ has full rank, and $b = M_{\vert_{X_0}}c$. Then for each $\epsilon>0$, there exist $\tau, \delta>0$ with the following property: If $X = \set{x_1', \dots, x_s'} \cup \widetilde{X}$ with
    \begin{align*}
        \abs{x_i -x_i'} &< \delta, \ i =1, \dots, s \quad \inf_{\tilde{t} \in\widetilde{X}, i=1, \dots s} \abs{x-x_i} &\geq \tau,
    \end{align*}
    such that
    \begin{align*}
            \left(M\vert_X\right)^{\dagger} b = \sum_{i=1}^s \alpha \delta_{x_i'} + \sum_{\tilde{x} \in \widetilde{X}} \beta_{\tilde{x}} \delta_{\tilde{x}}
    \end{align*}
    with
    \begin{align*}
        \frac{\norm{\alpha - c}_2}{\norm{c}_2} < \epsilon, \ \frac{\norm{\beta}_2}{\norm{c}_2} <\epsilon
    \end{align*}
\end{prop}

The arguments to obtain this result are relatively standard. For completeness, we include them in the appendix.

\subsection{Simulation Results in $1D$}

Now, we carry out the numerical experiments to compare the two strategies described above. In all experiments, we  assume $s=3$ thermal sources inducing at time $t=0$ along with a bar with the length $2 \pi$. The amplitudes of all sources are assumed to $\{c_n\}_{n=1}^3=1$, for simplicity. The grid-line source distribution for comparing methods is assigned as, $P=128$.

We fix the number of sensors $N_s=16$ for all experiments. Note that this is much smaller than $P=128$. We furthermore only take one time sample, say at $\tau =  \rho/\Delta_2^2$, as was done in the theoretical section.  Accordingly, $\Delta_2 = \tfrac{2 \pi}{N_S \cdot 1}$. The parameter $\rho$ is chosen as $\rho = \tfrac{\rho_{MAX} + \rho_{MIN}}{2} $, so that it obeys the bound \eqref{eq:16}. We then perform two experiments: First, we generate data from assigning source positions on the grid. Then, we instead choose them off the grid. Then, we compare the two methods in the case that the measurements are contaminated by Gaussian noise of strength $40~B$. The last experiment is also repeated for $\rho = 5\tfrac{\rho_{MAX} + \rho_{MIN}}{2}$, which clearly violates the bound \eqref{eq:16}.

Note that the method from \cite{ranieri2011sampling} assumes that the sources are on the grid, and is only guaranteed to work when $\rho$ obeys the bound \eqref{eq:16}(at least in the noisy case). We still choose to test its performance for cases when these two assumptions are not met, in order to demonstrate possible improvements using our method.

As is suggested in \cite{ranieri2011sampling}, we apply the smoothed $\ell_0$ (SL0) method \cite{mohimani2009fast}  to solve the inverse problem $b=M \mu_0$. For our proposed scheme, we choose the threshold $\tau_{peak}$ $1-\tfrac{0.5}{4^k}$ for iteration $k$, and the stopping criterion $p^{k+1} - p^{k} < 10 ^{-4}$, where $p^k$ is the optimal value of the $k$:th optimization problem $(\calP_{dual})$. For the optimization, we use cvx, a package for specifying and solving convex programs \cite{cvx1, cvx2}.

\subsubsection{Results for Noiseless Measurements}

The results for the experiments are presented in Figure \ref{Result1}. As can be seen, the method proposed in this paper almost perfectly recovers the peaks, both in the on-grid as well as in the off-grid case, whereas the peaks of the signal recovered by the method from \cite{ranieri2011sampling} are 'smeared out', making exact localization of the positions hard. This furthermore leads to the amplitudes being incorrectly recovered. Note that the performance of our method is similar in the on- and off-grid case, whereas the performance deteriorates slightly when going off-grid for the one proposed in \cite{ranieri2011sampling}. In particular note the tendency to a spurious peak at $x\approx 3.8$ for the latter method in the off-grid case.

\begin{figure}[t]
\begin{center}
     \includegraphics[width=90mm]{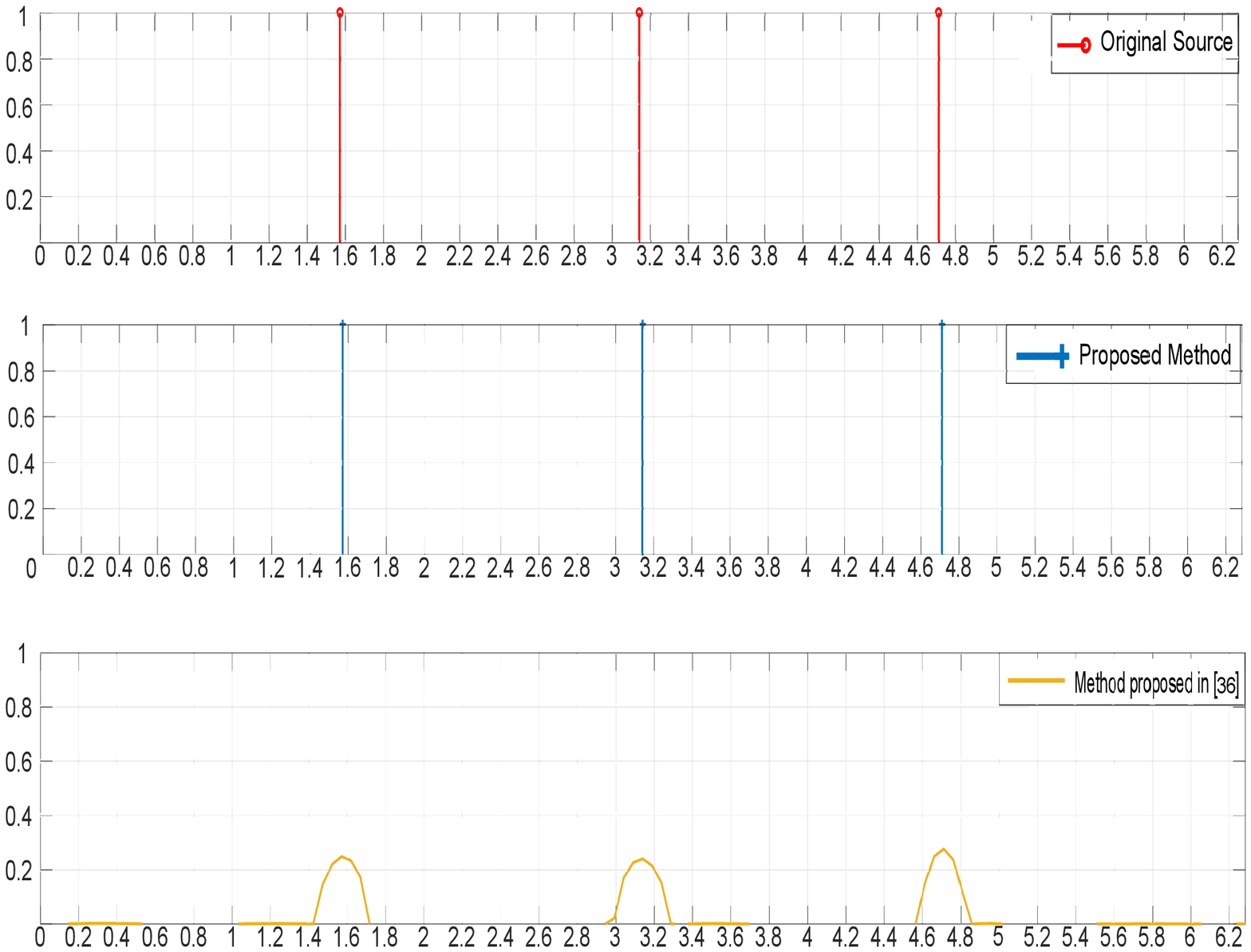} 
     \includegraphics[width=90mm]{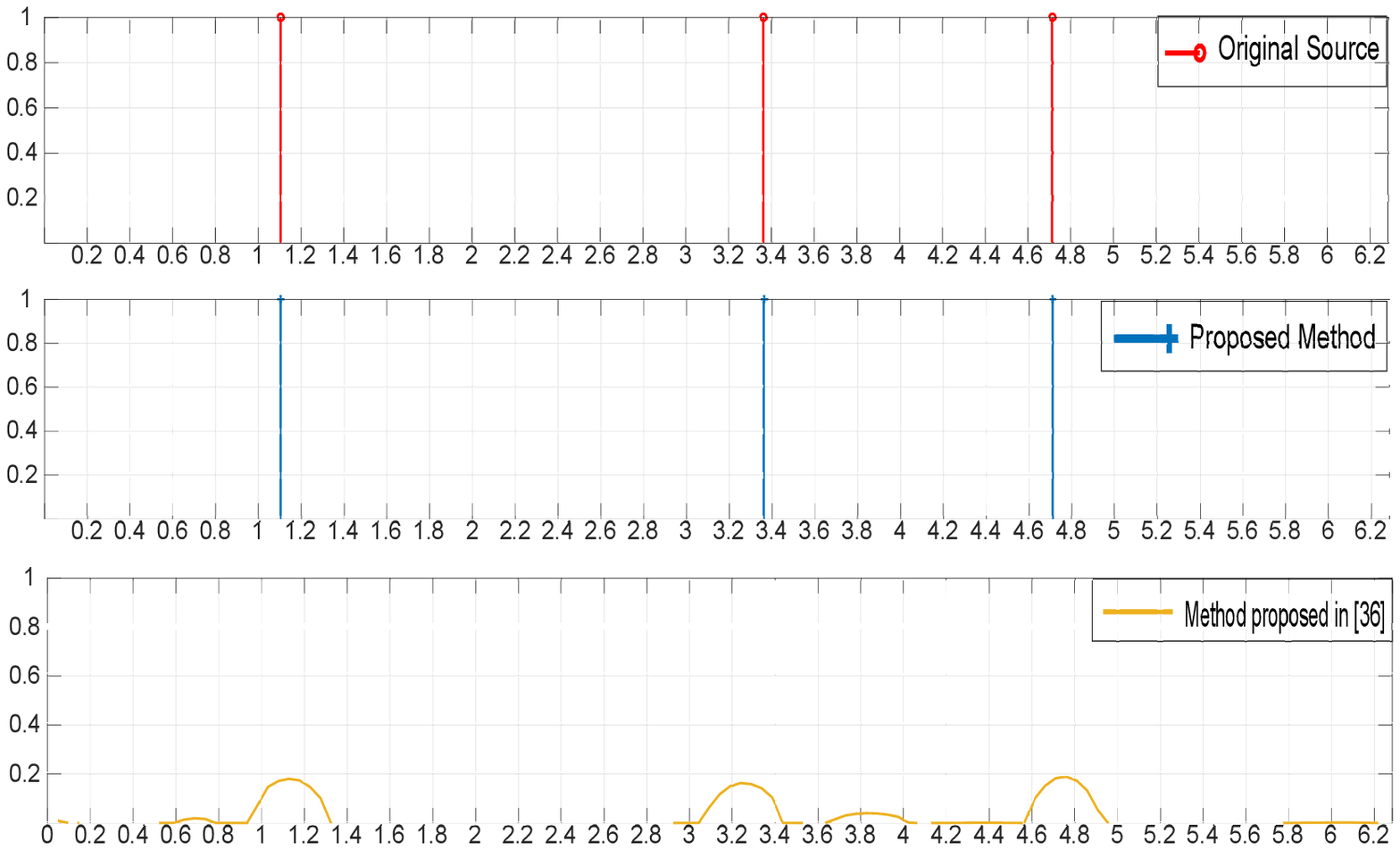}   
\end{center}

    \caption{\label{Result1} The source localization performance results in the noise-less case. The above figure shows the on-grid performance, the below figure the off-grid performance.}

\end{figure}


The fact that our method is robust against violating the  bound \eqref{eq:16} for well-conditioning the system dictionary matrix can also provide a significant impact in scenarios where one is willing to monitor the heat induction close to the beginning of a thermal activity. In these settings, the sensing machinery including time and space densities does not have to obey any bound, and one can monitor the target activity with the desired frequency and time sampling density. 

\subsubsection{Results for Noisy Measurements}

We now compare the performance of the introduced method in presence of noise. It should be mentioned that the noisy scenario is not addressed in \cite{ranieri2011sampling}. This is not surprising, considering the highly ill-conditioned properties of the system model matrix, $M$. Still, we use this method as a comparison of performance in presence of noise. 

We manually tune the parameter $\lambda$ in $(\calP_{LASSO})$ to $ 0.35 \times s_d^2$ for our simulation procedure, where $s_d^2$ is the variance of Gaussian white additive noise showing its power. We also set smaller
stopping criteria, $p^{k+1} - p^{k} < 10 ^{-6}$,  in this simulation, for having more accuracy compared to the noiseless simulations.

Figure \ref{Result3} demonstrate the results of the two methods in noisy settings when the signal-to-noise ratio (SNR) is equal to $40~dB$. As can be seen, the method proposed in  \cite{ranieri2011sampling}  completely fails, both when $\rho$ is within, as well as without, the  bounds given by \eqref{eq:16}. Note that in the case that $\rho$ violates \eqref{eq:16}, the outputted signal has amplitudes in a range of $2 \times 10^{10}$, suggesting that the reason for the failure is the ill-conditioning of $M$.  
The method we propose, on the other hand, in both cases delivers a very sharp estimate of both the amplitudes as well as source locations as well.

\begin{figure}[h]
\begin{center}
    \includegraphics[width=90mm]{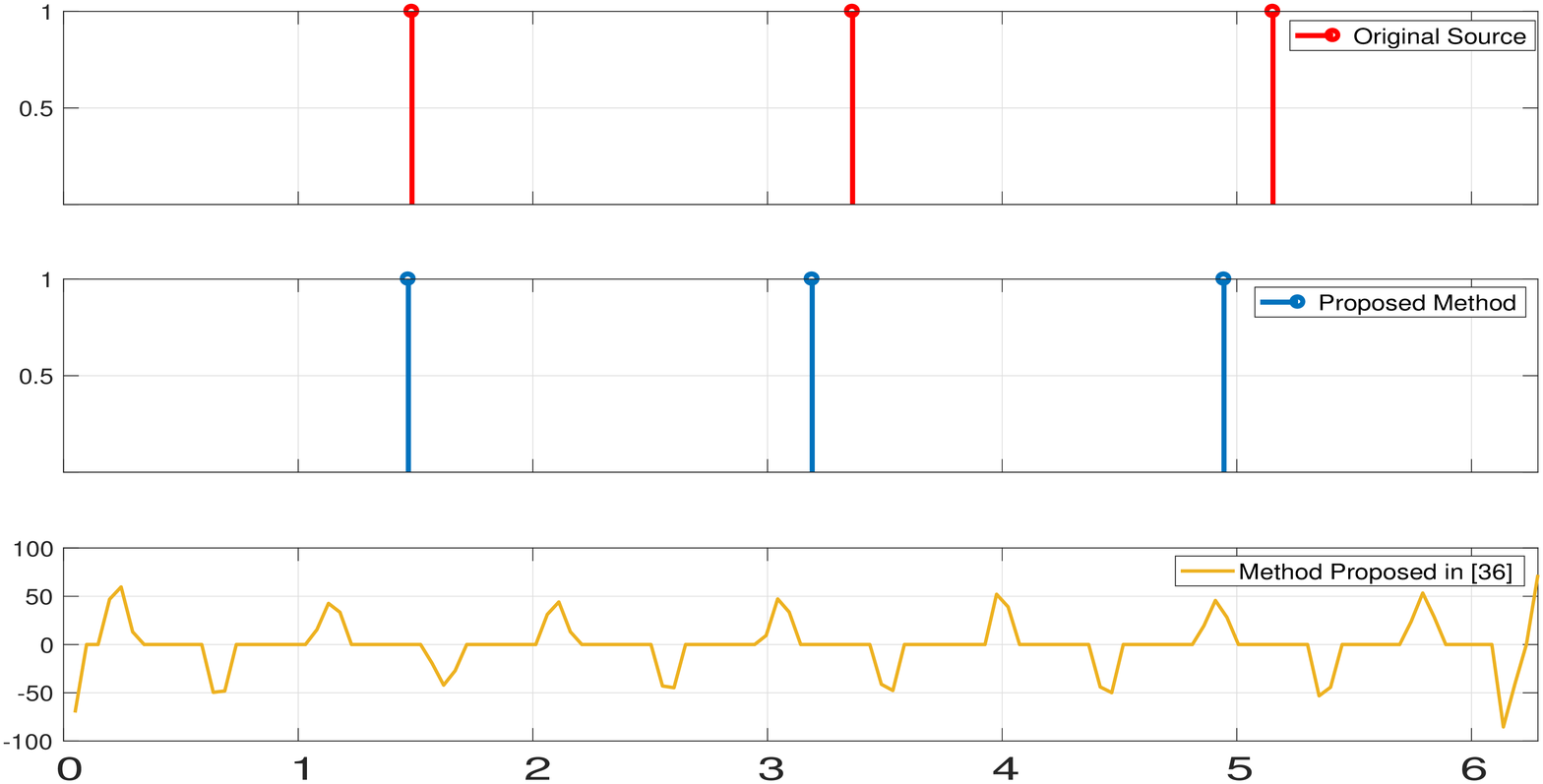}\\
    \includegraphics[width=90mm]{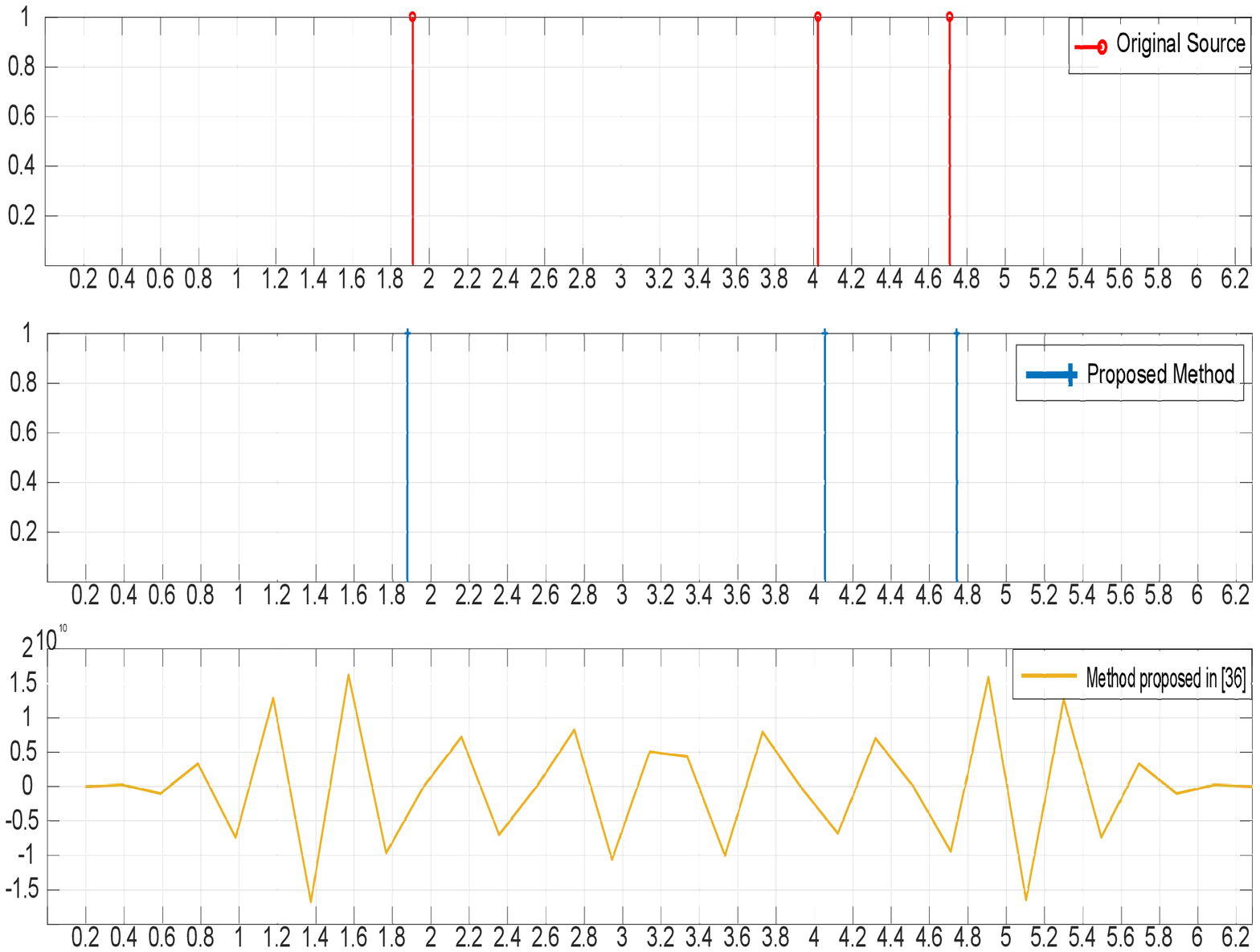}   
\end{center}
    \caption{\label{Result3} The source localization performance results in a noisy scenario with $SNR=40dB$. In the above figure, $\rho$ is chosen within the bounds described by \eqref{eq:16}, in the below figure, $\rho$ violates $\eqref{eq:16}$.}
\end{figure}

\subsection{$2D$ Simulation}

Let us finally test the performance of the method in the case of two dimensions. From a mathematical point of view, the procedure is exactly as before. From a practical point of view, we make one adjustment: Due to the increase of number of points in the grid, performing the optimization using explicit matrices and a standardized software like cvx is not feasible. Instead, we implement the primal-dual method \cite{chambolle2011first}, which only requires the operations $M$ and $M^*$ to be implemented as operators. 

In order to find the points where the final dual certificate has absolute value $1$, we used the following heuristic: First, we recorded all points $p_i$ in which the certificate $\nu_{final}$ had an absolute value larger than a certain threshold. This threshold was manually tuned to optimize the performance. These points where typically clustered around the peaks (this is due to the design of our method: close to the peaks, a lot of points are added during the iterations). Therefore, we performed a $k$-means clustering to decide $3$ initial positions of the peaks. Using these initial positions, we then performed a gradient descent to find the local maxima of $\abs{\nu_{final}}$, which we finally used as our final grid $X_*$.

\begin{figure*}[t]
\begin{minipage}[b]{.33\linewidth}
  \centering
  \centerline{\includegraphics[width=6cm]{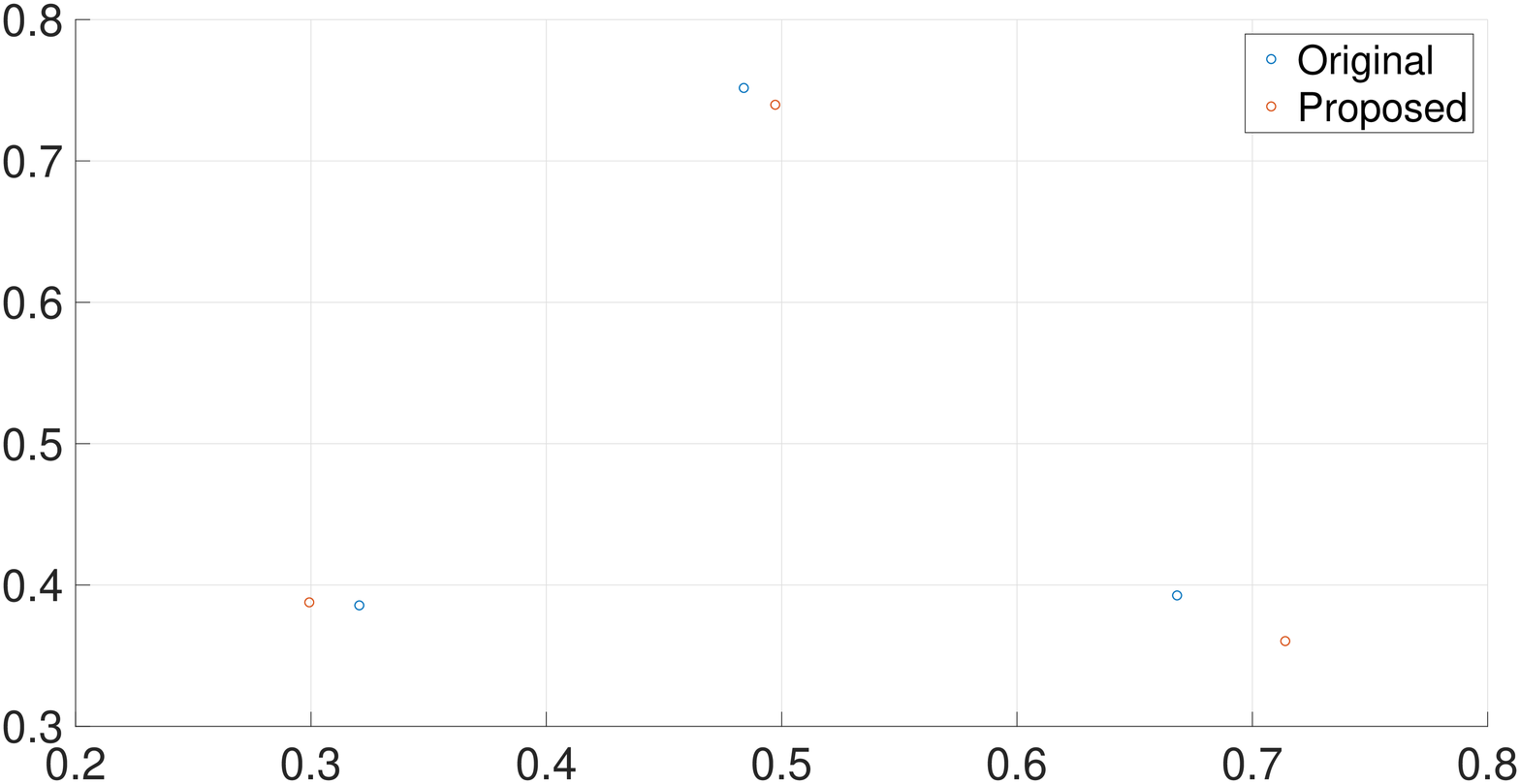}}
  \centerline{(a)}\medskip
\end{minipage}
\begin{minipage}[b]{0.33\linewidth}
  \centering
  \centerline{\includegraphics[width=6cm]{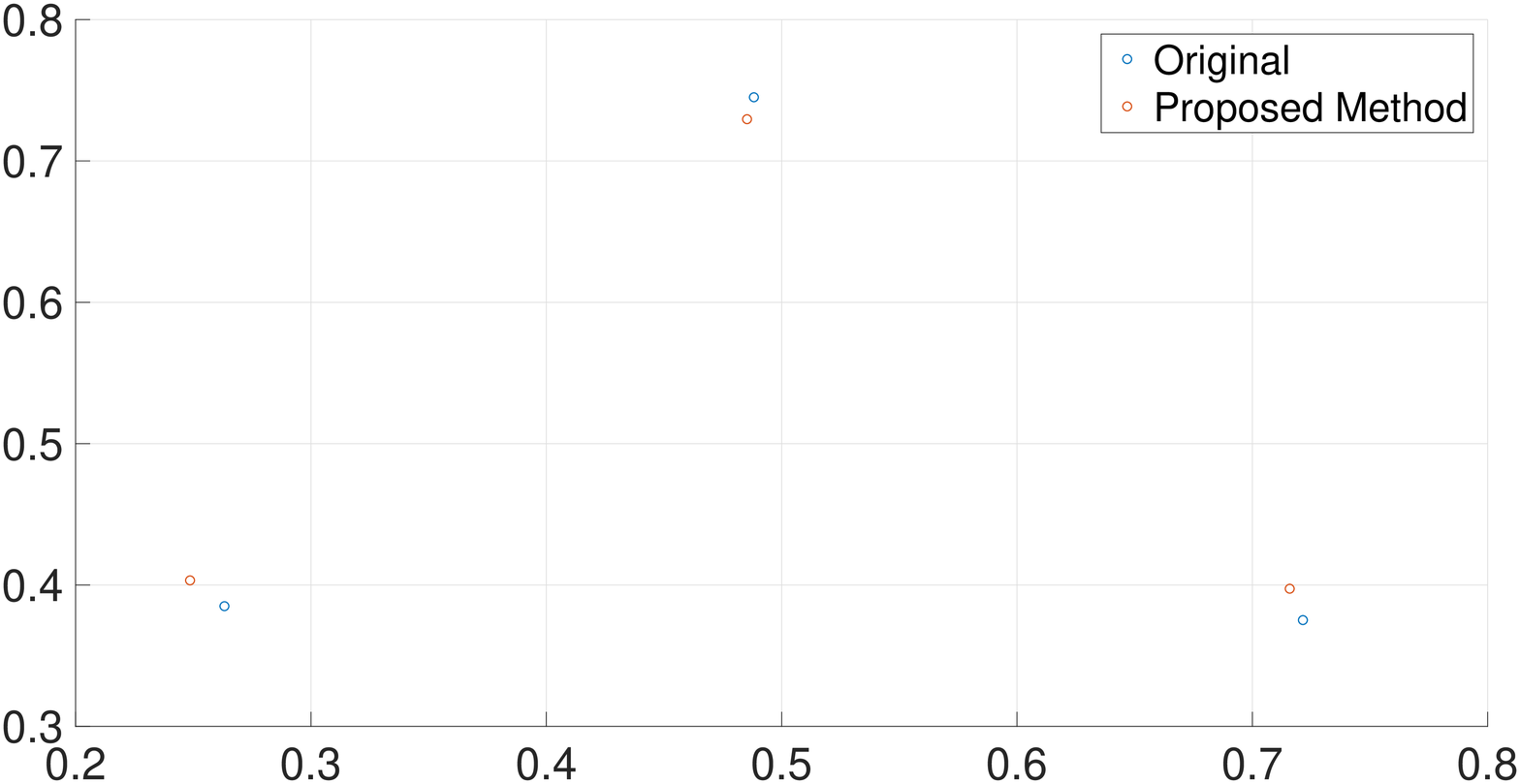}}
  \centerline{(b)}\medskip
\end{minipage}
\begin{minipage}[b]{0.33\linewidth}
  \centering
  \centerline{\includegraphics[width=6cm]{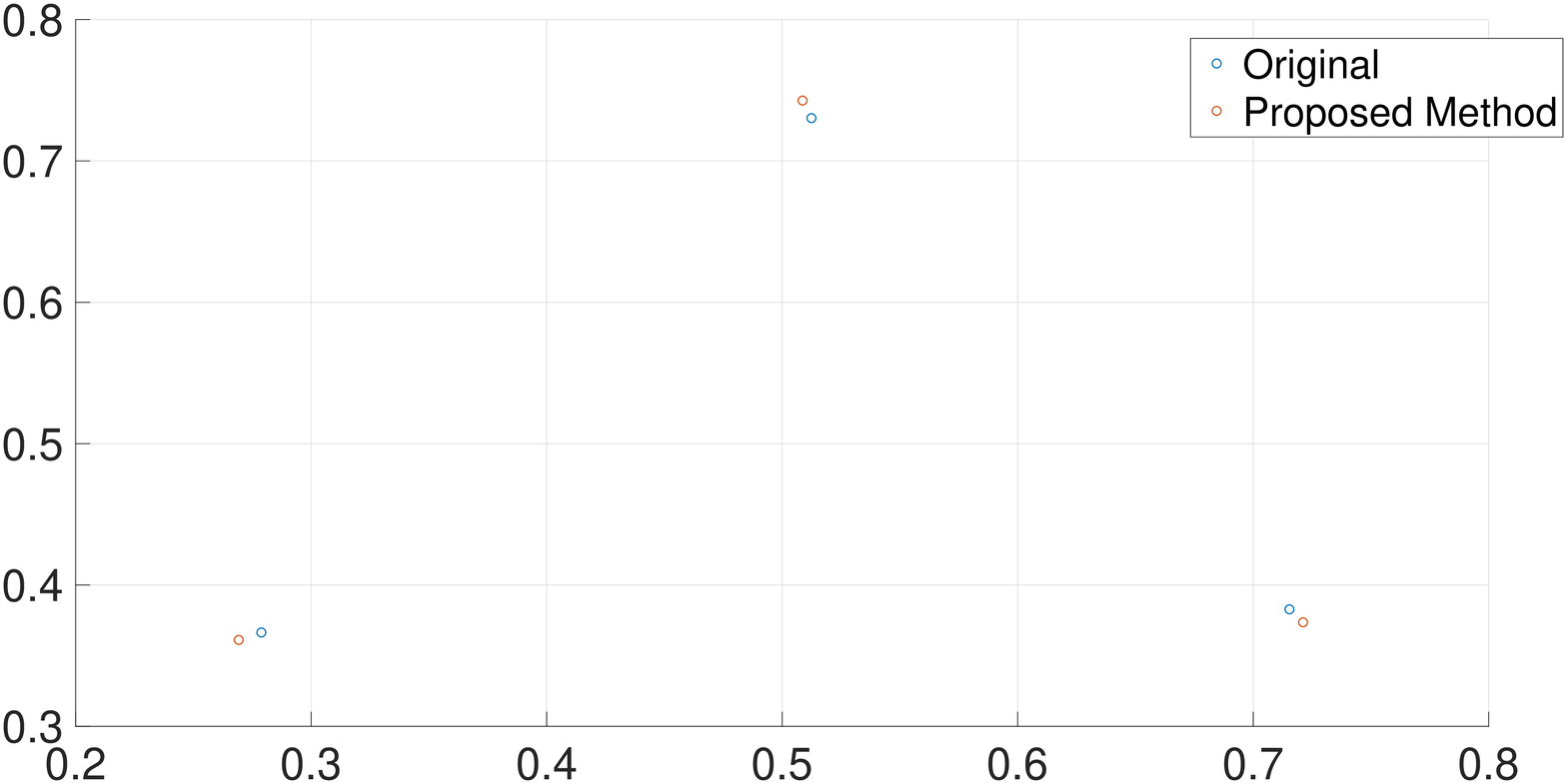}}
  \centerline{(c)}\medskip
\end{minipage}

\caption{Visualization of obtained locations in different SNR settings. (a) $SNR=0~dB$, (b) (a) $SNR=20~dB$, (c)  $SNR=30~dB$.}
\label{Fig:Vis}
\end{figure*}

In Figure \ref{Fig:Vis}, we plot the reconstructed source localizations, compared to the ground truth ones, for different noise levels. We see that the performance is good, and it gracefully tackles the noise, producing reasonable results even for low SNR regimes like $0~dB$.

Figure \ref{3DShape} demonstrates the amplitude of the reconstructed sources compared to the original ones, in a low SNR setting, $SNR=~0dB$.

\begin{figure}[t]
\begin{center}
    \includegraphics[width=90mm]{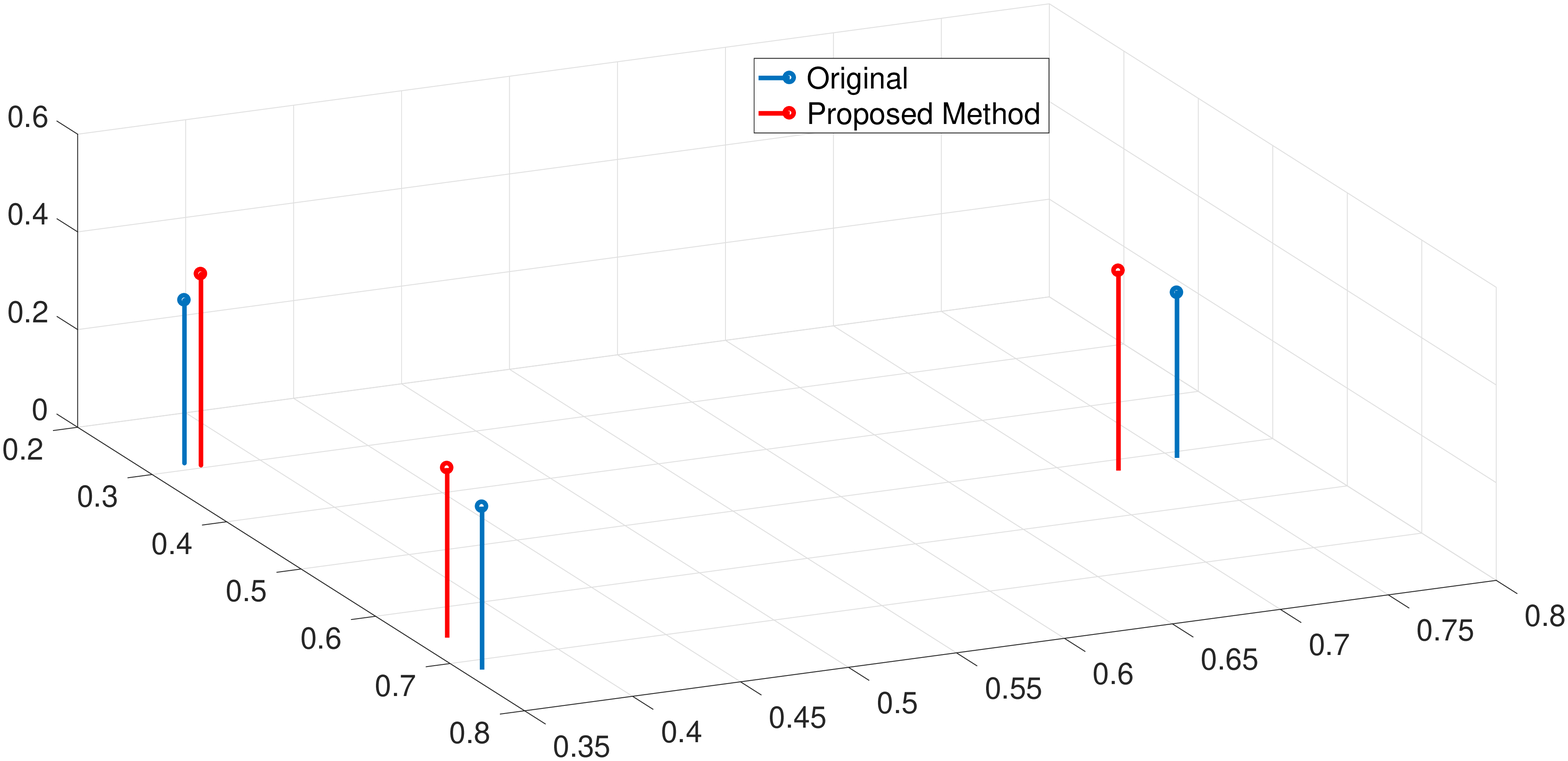}\\
    \includegraphics[width=90mm]{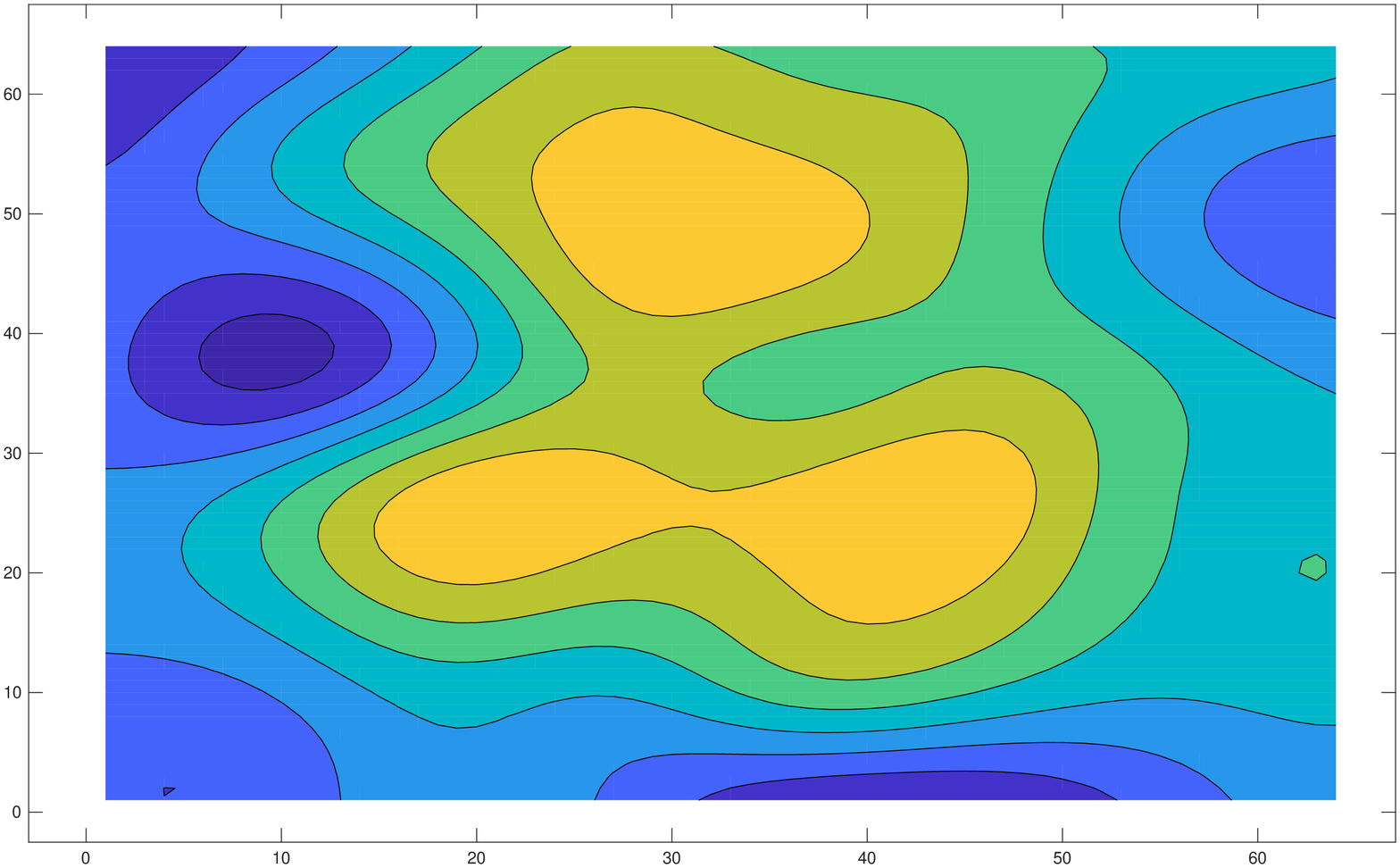}   
\end{center}
    \caption{ Reconstructed amplitudes (upper) and thermal field (below) in low SNR regime, $SNR=~0dB$} \label{3DShape}
\end{figure}


%
%

%
%

\section{Conclusion}

In this paper we addressed the thermal source localization problem which has a wide range of applications and has attracted a lot of attention, recently. Based on infinite-dimensional compressed sensing theory, we proposed a method relying on TV-norm minimization for recovering the amplitudes and locations of the thermal sources in a 2D scenario. Using the soft recovery framework from \cite{flinth2017soft}, we provided rigorous mathematical guarantees, showing that the placements of the thermal sources will be approximately recovered by our method. In fact, we even extend the mentioned framework to include noisy measurements, which has an interest on its own

Based on the theoretical analysis, we have proposed a heuristic scheme in order to implement the infinite dimensional setting, and performed numerical simulations in a different circumstances for validating our methods. The numerical results demonstrate significantly increased performance and robustness of our proposed method compared to classical algorithms. In fact, it can tackle several challenges and phenomenons in this area such as insufficient number of spatiotemporal samples, highly ill-conditioning property of Green function matrix specially in noisy environment, discretization error and off-grid sources positioning issue.

%
%
%
%

%

\subsection*{Acknowledgements}
Axel Flinth acknowledges support from the Deutsche Forschungsgemeinschaft (DFG) Grant KU 1446/18-1. He also wishes to thank Jackie Ma and Philipp Petersen for interesting discussions on this and other topics.

Ali Hashemi acknowledges support from Berlin International Graduate School in Model and Simulation based Research (BIMoS). He also would like to thank Saeid Haghighatshoar and Ngai-Man Cheung for interesting discussions on CS applications for source localization inverse problem.

Both authors also acknowledge support from the Berlin Mathematical School (BMS). They further wish to thank Gitta Kutyniok for careful proofreading and valuable suggestions for improving the paper.

\bibliographystyle{abbrv}
\bibliography{bibliographyCSandFriends}

\section{Appendix}

The appendix consists of two parts. The first one provides the left out proofs for proving Theorems \ref{th:main}. In the second one, we extend the framework from \cite{flinth2017soft} to include also noisy measurements.

\subsection{Miscellaneous Results needed for the Proof of the Main Result}
\subsubsection{Proof of Lemma \ref{lem:MProps}} \label{sec:Mcont}

        We have for $x,t \in \calS$ arbitrary
        \begin{align*}
\int_{\R^2} G(p-x,t)d\mu(p) &= \int_{\R^2} \widehat{G}(\xi -x, t) \overline{\widehat{\mu}(\xi)} d\xi = \int_{\R^2} G(\xi-x,t) \widehat{\phi}(\xi)^{-1} \widehat{\phi}(\xi)\overline{\widehat{\mu}(\xi)} d\xi\\&
\leq  \left(\int_{\R^2}\abs{\widehat{G}(\xi-x,t)\widehat{\phi}(\xi)^{-1}}^2 \right)^{1/2}\norm{\mu}_\calE,
        \end{align*}
        where we used Cauchy-Schwarz in the last step. Now we have
        \begin{align*}
              \widehat{G}(\xi-x,t)\widehat{\phi}(\xi)^{-1}  = (2\pi)^{-1} t\Lambda\exp( -t\abs{\xi}^2/2) \exp(\Lambda\xi^2/2) \exp(i x\xi),
        \end{align*}
        which is square integrable if and only if $t>\Lambda$. Since $\int_{\R^2} G(p-x,t)d\mu(p) = (M\mu)_{x,t}$ and $\abs{\calS}<\infty$, we obtain $\norm{M\mu}_2 \leqsim \norm{\mu}_\calE$ if all $t>\Lambda$.
        
        To calculate the adjoint operator, let $\lambda \in \C^m$ and $\upsilon \in \calE$ be arbitrary. We have
        \begin{align*}
            \sprod{M\upsilon, \lambda} 
            = \sum_{(x,t) \in \calS} \int_{\R^2} G(p-x,t)d\upsilon(p) \overline{\lambda}_{x,t} &= \int_{\R^2} \left( \sum_{(x,t) \in \calS} \overline{\lambda}_{x,t} \widehat{G}(\xi-x,t)\abs{\widehat{\phi}}^{-2}\right)\widehat{\upsilon} \abs{\widehat{\phi}}^2 d\xi \\
            &= \sprod{\upsilon, \sum_{(x,t) \in \calS} \lambda_{x,t} \widetilde{G}(p-x,t)}_\calE.
        \end{align*}
        Note that in particular $\widetilde{G} \in \calE$ for $t > \Lambda$, since $\hatphi \cdot \calF \widetilde{G} = \overline{\hatphi}^{-1}\widehat{G}$ is square integrable for $t>\Lambda$. \hfill $\blacksquare$

   \subsubsection{Approximating Gaussians with Gaussians} \label{sec:Appr}
   
   As was advertised in the main body of the text, the problem of approximating Gaussians with Gaussians can be boiled down to approximating a plain wave $\exp(i\Delta \cdot \omega)$ using trigonometric polynomials. Since most literature analyzing this problem is treating the one-dimensional case, let us carry out the argument: First, for $p \in \N$, define the \emph{Jackson Kernel}
\begin{align*}
    J_p(x) = \lambda_p \left(\frac{\sin(px/2)}{\sin(x/2)}\right)^4, \quad x \in \R,
\end{align*}
where $\lambda_p$ is chosen so that $\int_{-\pi}^\pi J_p(x)dx = 1$. One can prove \cite[Ch. 4, Sec. 2; p. 55]{Lorentz1966} 
\begin{align} \label{eq:JacksonProps}
  \int_{-\pi}^\pi \abs{t} J_p(t) dt \leqsim p^{-1}.
\end{align}
It is furthermore clear that if $g$ is a $2\pi$-periodic function on $\R^2$, then
\begin{align} \label{eq:JacksonConv}
    &g *(J_p \otimes J_p) (\omega_1, \omega_2) \\ \nonumber
    &\qquad = \int_{[-\pi,\pi]^2} g(\xi_1, \xi_2) J_p(\omega_1 -\xi_1) J_p(\omega_2- \xi_2) d\xi_1 d\xi_2
\end{align}
is  a sum of complex exponential $\exp(i n \cdot \omega)$ with $\norm{n}_\infty \leq 2p$. With this toolbox at hand, we may deduce the following:
\begin{lem} \label{lem:ExpAppr}
    Let $p \in \N$. There exist complex scalars $c_n$ such that
    \begin{align} \label{eq:errBound}
       \sup_{\omega \in [-\pi/2, \pi/2]^2}   \abss{   \exp(i  \Delta \cdot \omega) -\sum_{ \norm{n}_\infty \leq 2p} c_n\exp(i n \cdot \omega)} \leqsim p^{-1}
    \end{align}
    and
    \begin{align} \label{eq:absBound}
        \abss{\sum_{\norm{n}_\infty \leq 2p} c_n\exp(i n \cdot \omega)} \leq 1 .
    \end{align}
    The $c_n$ furthermore obeys the inequality
    \begin{align} \label{eq:normBound}
        \sum_{\norm{n}_\infty \leq 2p} \abs{c_n}^2 \leq 1
    \end{align}
\end{lem}

\begin{proof}
    First, let $g$ be a $2\pi$-periodic, continuously differentiable function such that $g(\omega) = \exp(i \Delta \omega)$ for all $\omega \in [-\pi/2,\pi/2]^2$, and furthermore have the property that $\abs{g}$ is uniformally bounded by one and Lipschitz-continous with constant $1$.  We may construct such a function by multiplying two univariate functions $a_1, a_2$, where each $a_j$ is defined as $\exp(i \Delta_j \omega_j)$ for $\pi\in [-\pi/2,\pi/2]$, linearly prolonged between $\pi/2$ and $3\pi/2$, and then extended $2\pi$-periodically on the whole real line. To be concrete:
    \begin{align*}
        a_j(\theta) = \begin{cases}
        e^{i\Delta_j \theta} &, \theta \in [-\pi/2, \pi/2] \\
        e^{\pi i \Delta_j /2} + \left(\theta - \frac{\pi}{2}\right)\lambda_j &, \theta \in [\pi/2, 3\pi/2]
        \end{cases} , \quad 
        \lambda_j =  \frac{e^{3\pi i  \Delta_j/2}- e^{\pi i \Delta_j/2}}{\pi}.
    \end{align*}
    $a_j$  are then in modulus bounded by $1$ and have Lipschitz constants $L_j$ bounded by $\Delta_j \leq 1/2$, which proves that the Lipschitz constant of $a_1 \cdot a_2$ is smaller than $L_1 \norm{a_1}_\infty + L_2 \norm{a_2}_\infty\leq 1$. 
    
    We now define our sum of complex exponentials as \eqref{eq:JacksonConv}. The fact that $\abs{g} \leq 1$ and the normalization of $J_p$ then implies \eqref{eq:absBound}. To see that also \eqref{eq:errBound} is true, notice that for $\omega \in [-\pi/2, \pi/2]$, we have $\exp(i \Delta \omega) = g(\omega)$. Consequently, for such $\omega$,
    \begin{align*}
        \abss{\exp(i  \Delta \cdot \omega) -\sum_{n\in \set{-\tfrac{m}{2}, \dots, \tfrac{m}{2}}^2} c_n\exp(i n \cdot \omega)} &
      =  \abss{\int_{-\pi}^\pi \int_{-\pi}^\pi \left(g(\omega-\theta)-g(\omega)\right)J_p(\theta_1)J_p(\theta_2) d\theta_1 d\theta_2} \\
       & \quad \leq \int_{[-\pi,\pi]^2} (\abs{\theta_1} + \abs{\theta_2})J_p(\theta_1)J_p(\theta_2) d\theta  \leqsim p^{-1},
    \end{align*}
    where we used that the Lipschitz constant of $g$ is smaller than $1$ and \eqref{eq:JacksonProps}.
    
    To deduce \eqref{eq:normBound}, we note that the orthogonality of the complex exponentials on the rectangle $[-\pi,\pi]^2$ gives
    \begin{align*}
        \sum_{\norm{n}_\infty \leq 2p} 4\abs{c_n}^2 \pi^2= \norm{ \sum_{\norm{n}_\infty \leq 2p} c_n e^{in \omega}}_2^2 &=\norm{g * (J_p \otimes J_p)}_2^2 \leq \norm{g}_2^2 \norm{J_p \otimes J_p}_1^2  \\
        &\leq 4\pi^2\norm{g}_\infty^2 \norm{J_p \otimes J_p}_1^2 \leq 4\pi^2
    \end{align*}
    In the final steps, we applied Young's inequality and again the uniform bound on $g$. 
\end{proof}

Now, we may easily extract Lemma \ref{cor:scalars}.

\begin{proof}[Proof of Lemma \ref{cor:scalars}] 
Let us choose $c_n$ as in Lemma \ref{lem:ExpAppr}. According to the prior discussion, the entity we need to estimate is the final integral in Equation \eqref{eq:FourierTrick}. Let us split the integral into two parts: one over $I=[-\pi/2,\pi/2]^2$ and one over the rest of $\R^2$. For the first integral, due to \eqref{eq:errBound}, we then have 
\begin{align*}
    m^2 \abss{ \int_{I} \widehat{\psi}(m\omega) \left( \exp(i  \Delta \cdot \omega) -\sum_{n} c_n\exp(i n \cdot \omega)\right) d\omega} 
    \leqsim  m^{-1}\cdot m^2 \abss{ \int_{\R^2} \widehat{\psi}(m\omega) d\omega} \leqsim m^{-1} .
\end{align*}
As for the second, due to \eqref{eq:absBound}, we have 
\begin{align*}
    m^2 \abss{ \int_{\R^2 \backslash I} \widehat{\psi}(m\omega) \left( \exp(i  \Delta \cdot \omega) -\sum_{n} c_n\exp(i n \cdot \omega)\right) d\omega} 
    \leq   2m^2 \int_{\R^2 \backslash I} \abs{\widehat{\psi}(m\omega)} d\omega 
    \leq 2\int_{ \R^2 \backslash mI}  \abs{\widehat{\psi}(\omega)} d\omega
\end{align*}
In our case, $\psi(x) = \exp(-\abs{x}^2/(2\Lambda))$, so $\widehat{\psi}(\omega) \sim \Lambda^{-1} \exp(-\Lambda \abs{\omega}^2/2)$, which has the consequence
\begin{align*}
     \int_{ \R^2 \backslash mI}  \abs{\widehat{\psi}(\omega)} d\omega &\leqsim \int_{ \R^2 \backslash  m I} \exp(-\Lambda\abs{\omega}^2/2) d\omega 
     \leqsim \exp(-m^2\Lambda/2) \leqsim (\sqrt{\Lambda} m)^{-1},
\end{align*}
where the last inequality follows from the inequality
\begin{align}
    \sup_{x \in \R} xe^{-x^2} \leq \sqrt{2}^{-1} e^{-\tfrac{1}{2}},
\end{align}
which can be proved via elementary calculus (we simply need to maximize the function $xe^{-x^2}$ over the reals).
\end{proof}

\subsubsection{Proof of Proposition \ref{prop:MoorePenrose}} Finally, we prove the statement about using the Moore-Penrose inverse of the restricted operator $M\vert_T$ to recover the amplitudes of the ground truth signal.

\begin{proof}[Proof of Proposition \ref{prop:MoorePenrose}]
Let us denote $T'= \set{t_1', \dots t_s'}$. Then we have
\begin{align*}
     M\vert_{T_0} c = M\vert_{T'} \alpha + M\vert_{\widetilde{T}} \beta.
\end{align*}
We now aim to prove that $\alpha \approx c$ and $\beta \approx 0$.

The $i$:th column of the matrix $M\vert_{T'}$ is given by $M(\delta_{t_i'})$. As $t_i' \to t_i$, the measure $\delta_{t_i'}$ converges to $\delta_{t_i}$ in the weak--$*$--topology. Due to continuity of $M$, this has the consequence $M(\delta_{t_i'}) \to M(\delta_{t_i})$. Hence, 
    $M \vert_{T'}$ and  $M \vert_{T_0}$  will be close if $\sup_i \abs{t_i - t_i'}<\delta$. Since $M\vert_{T_0}$ has full rank, $M\vert_{T'}$ will in particular also have full rank for such values of $t_i$. This implies that
    \begin{align*}
        ((M\vert_{T_0})^+ M\vert_{T'})^+ = M\vert_{T'}^+ M\vert_{T_0}
    \end{align*}
    The full rank of $M\vert_{T_0}$ also implies that $(M\vert_{T_0})^+ M\vert_{T_0}= \id$, so that $(M\vert_{T_0})^+M_{T'}$ also is close to $\id$. Since the operation of forming the pseudoinverse is continuous in $\id$, we finally conclude that if $\delta$ is chosen small, $(M\vert_{T'})^+ M\vert_{T_0}$ is close to $\id$.
    Finally, due to the full rank of $M\vert_{T'}$, we have $(M\vert_{T'})^+ M\vert_{T'}=\id$. This has the consequence
    \begin{align}
        \alpha = (M\vert_{T'})^+ M\vert_{T'} \alpha = (M\vert_{T'})^+M\vert_{T_0}c - (M\vert_{T'})^+ M\vert_{\widetilde{T}} \beta.
    \end{align}
    If we now argue that $M\vert_{T'}^+M\vert_{\widetilde{T}}\beta$ is almost equal to $0$, we are done. First note that if $\norm{p-q}_2\geq \tau' := \tau-\epsilon$, all products of the form $G(p-x,t)G(q-x,t)$ will be small for all $x$ -- since in each case, at least one of the factors will be small. This has the consequence that
    \begin{align}
        \abss{\sprod{ M\delta_{p}, M\delta_q}}=  \abss{\sum_{x \in \calS} G(p-x,t)G(q-x,t)} \leq \delta',
    \end{align}
    where $\delta'$ is some small real number. Consequently, due to the separation between $\widetilde{T}$ and $T'$, the component of $M\vert_{\widetilde{T}}$ parallel to the column space of $M\vert_{T'}$ will be almost zero, say smaller than $\epsilon$ in spectral norm. Since $(M\vert_{T'})^+M\vert_{\widetilde{T}}\beta  =(M\vert_{T'})^+\Pi_{\langle M\vert_{T'}\rangle}M\vert_{\widetilde{T}}\beta $, where  $\Pi_{\langle M\vert_{T'}\rangle}$ denotes the  orthogonal projection onto the column space of $M\vert_{T'}$, \begin{align*} \norm{(M\vert_{T'})^+M\vert_{\widetilde{T}}\beta} \leq& \norm{(M\vert_{T'})^+} \norm{\Pi_{\langle M\vert_{T'}\rangle}M\vert_{\widetilde{T}}\beta} \leq \epsilon  \norm{(M\vert_{T'})^+}\norm{\beta}_2 \end{align*}
    which was to be proven.
\end{proof}

\subsection{Extending Soft Recovery to Noisy Measurements} \label{app:soft}

In this appendix, we will extend the framework of \cite{flinth2017soft} to a noisy setting. We will first stay in the very general setup of the mentioned paper, and then specialize our results to the special case of recovery of sparse atomic measures using the $TV$-norm.

We begin by describing the setting which was considered in \cite{flinth2017soft}. We will stay relatively streamlined and refer to the mentioned publication for details. Let $I$ be a locally compact and separable metric space, and $\calH$ a Hilbert space over $\K$, where $\K$ denotes  either $\R$ or $\C$. We call a family $(\varphi_x)_{x\in I} \sse \calH$ a \emph{dictionary} if the map
\begin{align*}
    D^* : \calH \to \calC(I), v \mapsto \left(x \mapsto \sprod{v, \vphi_x}\right)
\end{align*}
is continuous and bounded. $\calC(I)$ thereby denotes the space of continuous functions on $I$ vanishing at infinity, equipped with the supremum norm. (A function is said to vanish at infinity if there for every $\epsilon>0$ exists a compact set $K \sse I$ with the property that $\abs{\varphi(x)} < \epsilon$ for $x \notin K$). $\calC(I)$ can be canonically identified with the dual of $\calM(I)$.

If $(\vphi_x)_{x\in I}$ is a dictionary, one can define a \emph{dictionary operator} $D : \calM(I) \mapsto \calH$ through duality:
\begin{align*}
    \forall v \in \calH: \sprod{v, D\mu} = \int_I \sprod{v, \vphi_x} d\mu(x).
\end{align*}
$D\mu$ should intuitively be thought of as an integral $\int_I \vphi_x d\mu(x)$.

The soft recovery framework now concerns the recovery of \emph{structured signals} $v_0$ by using \emph{atomic norm minimization}. We thereby call a signal \emph{structured} if it is of the form
\begin{align*} 
    v_0 = \sum_{i=1}^s c_{i} \vphi_{x_i} = D\left(\sum_{i=1}^s c_i \delta_{x_i} \right)
\end{align*}
for some $x_i \in I$ and $c_i \in \K$. The atomic norm $\norm{v}_\calA$ of an element $v \in \ran D$ is defined as (see \cite{flinth2017soft,chandrasekaran2012convex}) the optimal value of the minimization program
\begin{align*}
    \min \norm{\mu}_{TV}~\st~D\mu = v.
\end{align*}
The above minimization program always has a minimizer \cite{flinth2017soft}.

In fact, the framework holds more generally for signals of the form
\begin{align}
   v_0 = c_{x_0} \delta_{x_0} + D(\mu^c), \label{eq:structured}
\end{align}
where $\mu^c \in \calM$ is arbitrary. Intuitively, we may think of it as 'the rest of the peaks' $\sum_{x \neq x_0} c_x \delta_x$ in a structured signals, with possibly a bit of extra noise (compare to approximately sparse signals in the classical compressed sensing literature.)

The framework in \cite{flinth2017soft} provides a condition  guaranteeing the approximate recovery of a peak $x_0$ in the expansion of a structured signal $v_0$ as in \eqref{eq:structured} from linear measurements $b= Mv_0$. $M$ is thereby a continuous measurement operator from $\calH \to \K^d$. The condition reads as follows: Let $\sigma\geq 0$ and $\tau \in (0,1]$ be parameters. We call a vector $\nu \in \ran M^*$ a \emph{soft certificate for $x_0$} if
    \begin{align}
	\re \left(\int_I \sprod{\vphi_x, \nu} d(c_{x_0} \delta_{x_0} + \mu^c) \right) &\geq 1 \label{eq:Ankare} \\
	\abs{\sprod{\nu, \vphi_{x_0}}} &\leq \sigma \label{eq:atPoint} \\
	\sup_{x \in I} \abs{\sprod{\vphi_x, \Pi_{\sprod{\vphi_{x_0}}^\perp}\nu}} &\leq 1-\tau, \label{eq:orthCompSameSub}
\end{align}
where in the last equation, $\Pi_{\sprod{\vphi_{x_0}}^\perp}$ denotes the orthogonal projection onto the orthogonal projection of the span of $\vphi_{x_0}$. The main finding of \cite{flinth2017soft} is that for normalized dictionaries, the existence of a soft certificate for $x_0$ implies that all solutions $v$ of the program
\begin{align}
   \min \norm{v}_\calA~\st~Mv = b \tag{$\calP_\calA$}
\end{align}
where $b= Mv_0$ (with $v_0$ as in \eqref{eq:structured} with $\norm{v_0}_\calA =1$) has an expansion $v = D\mu_*$ for which $\sup_{x \in \supp \mu_*} \abs{\sprod{\vphi_x, \vphi_{x_0}}} \geq \tfrac{\tau}{\sigma}$. Hence, put a bit streamlined, there will be a peak  in the expansion of the minimizer of $\calP_\calA$ close to $x_0$.

We will now generalize this framework to cover also noisy measurement scenarios. Thus, we assume that $b = Mv_0 + e$ with $\norm{e}_2 \leq \epsilon$ and propose to use an atomic norm-version of the LASSO \cite{tibshirani1996regression}:
\begin{align}
    \min \norm{Mv - b}_2 ~\st~ \norm{v}_\calA \leq \rho, \tag{$\calP_\calA^{\rho,e}$}
\end{align}
where $\rho$ is a parameter. Ideally, one should choose $\rho= \norm{v_0}_\calA$ -- but since $\norm{v_0}_\calA$ in most scenarios is unknown, this parameter most often has to be calibrated.

We are now ready to state and prove the stability result.

\begin{theo} \label{th:SoftStable}
    Let $v_0$ be as in \eqref{eq:structured} with $\norm{v_0}_\calA=1$, $d\in \N$, $M: \calH \to \K^d$ be a continuous measurement operator,  and $\rho \geq 1$. Also let $b= Mv_0 + e$, where $\norm{e}_2 \leq \epsilon$.
    
    Suppose that there exists a soft certificate $\nu = M^*\lambda$ for $x_0$ with parameters $\sigma\geq 1$ and $\tau>0$. Then all minimizers $v_*$  of $(\calP_\calA^{\rho,e})$ have an atomic decomposition $v_* = D\mu_*$, $\norm{\mu_*}_{TV} = \norm{v_*}_\calA$ with
    \begin{align*}
        \sup_{x \in \supp \mu_*} \abs{ \sprod{\vphi_x, \vphi_{x_0}}} \geq \frac{\rho\tau- 2\norm{\lambda}_2 \epsilon +1-\rho}{\rho\sigma}.
    \end{align*}
\end{theo}

\begin{proof}
    Since $\rho\geq 1$, $v_0$ obeys the constraint of $\calP_\calA^{\rho,e}$. This immediately implies that the optimal value of it is smaller than $\epsilon$. Hence, if $v_*$ denotes a minimizer of $\calP_\calA^{\rho,e}$, we have
    \begin{align*}
        \norm{Mv_* -b}_2 \leq \epsilon.
    \end{align*}
    Since $b= Mv_0 + e$, we obtain
    \begin{align} 
        \epsilon \geq \norm{b-Mv_*}_2 &= \sup_{\kappa \in \K^d, \norm{\kappa}_2\leq 1} \re\left(\sprod{\kappa, Mv_0 - M v_* +e} \right) \nonumber\\
        &\geq \re \left(\sprod{\tfrac{\lambda}{\norm{\lambda}_2}, Mv_0 - M v_* + e} \right)
  \geq \re\left(\frac{1}{\norm{\lambda}_2}\sprod{\lambda, M(v_0-v_*)}\right) - \epsilon,\label{eq:rawBound}
        \end{align}
        where we used that $\norm{e}_2\leq \epsilon$. We now continue estimating $\re(\sprod{\lambda, M(v_0-v_*)}) = \re(\sprod{\nu, v_0-v_*})$. First, due to \eqref{eq:Ankare} and the definition of $D$, we have
        \begin{align}
            \re\left(\sprod{\nu, v_0}\right) &= \re\left(\sprod{\nu, D( c_{x_0} \delta_{x_0} + \mu^c)} \right)
            =\re\left( \int_I \sprod{\vphi_x, \nu} d(c_{x_0} \delta_{x_0} + \mu^c)\right) \geq 1. \label{eq:v0Bound}
        \end{align}
        Using \eqref{eq:atPoint} and \eqref{eq:orthCompSameSub}, we obtain
        \begin{align}
            \re\left(\sprod{\nu, v_*}\right) &= \re\left(\sprod{\nu, D\mu_*}\right) = \re\left(\int_{I} \sprod{\nu, \vphi_x} d\mu_*(x)\right) \nonumber\\
            &= \re\left(\int_{I} \sprod{\nu, \vphi_{x_0}}\sprod{\vphi_{x_0},\vphi_x} + \sprod{\nu,\Pi_{\sprod{\vphi_{x_0}}^\perp}\vphi_x} \ d\mu_*(x) \right)\nonumber \\& \leq \left(\sigma \cdot \sup_{x \in \supp \mu_*} \abs{\sprod{\vphi_{x_0},\vphi_x}} + (1-\tau)\right) \norm{\mu_*}_{TV}
           \leq  \left(\sigma \cdot \sup_{x \in \supp \mu_*} \abs{\sprod{\vphi_{x_0},\vphi_x}} + (1-\tau)\right)\rho, \label{eq:vStarBound}
        \end{align}
        where the final step follows from $\norm{\mu_*}_{TV}= \norm{v_*}_\calA \leq \rho$, since $v_*$ obeys the constraint of $\calP_{\calA}^{\rho,e}$. Hence, by inserting \eqref{eq:v0Bound} and \eqref{eq:vStarBound} into \eqref{eq:rawBound}, we obtain the inequality
        \begin{align}
            \epsilon \geq \frac{1}{\norm{\lambda}_2} \left( 1 - \rho \left(\sigma \cdot \sup_{x \in \supp \mu_*} \abs{\sprod{\vphi_{x_0},\vphi_x}} + (1-\tau)\right)\right) - \epsilon,
        \end{align}
        which  if rearranged yields the statement.
\end{proof}

After having proved a very abstract result true for any dictionary, we can now, exactly as in \cite{flinth2017soft}, very easily deduce Corollary \ref{cor:gStable}. The reason for this is that the $TV$-norm is in fact the atomic norm with respect to the dictionary $(\delta_p)_{p \in \R}$ in $\calE$ (for a proof of this fact, we refer to \cite[Section 4.3]{flinth2017soft}).

\begin{proof}[Proof of \ref{cor:gStable}]
Let $g = \calF^{-1} \abs{\widehat{\phi}}^2\calF \nu$. Since
 	\begin{align*}
 		\sprod{\nu, \delta_p}_\calE &= \sprod{\nu * \phi, \delta_p * \phi} = \int_{\R} \widehat{\nu}(t) \exp(-ipt) \abs{\widehat{\phi}}^2 dt \calF^{-1} \left( \widehat{\nu}\abs{\hatphi}^2\right)(p)=g(p), \\
 		\sprod{\delta_p, \delta_q}_\calE &= \int_{\R} \exp(ipt)\exp(-iqt) \abs{\widehat{\phi}}^2 dt = a(p-q),
 	\end{align*}
  the conditions on $g$ directly corresponds to the conditions on $\nu$ needed to apply Theorem \ref{theo:softRec}.

\end{proof}

\end{document}